\documentclass[12pt,a4paper]{amsart}

\usepackage{amsthm, amsfonts,amssymb,euscript}
\usepackage{amsmath}
\usepackage{amssymb}
\usepackage{amsthm}
\usepackage{graphicx}
\usepackage{amsfonts}
\usepackage{hyperref}
\usepackage{float}
\usepackage{slashed}

\numberwithin{equation}{section}

\usepackage{color}
\def\rr{{\Bbb R}}

\def\f12{\frac 1 2}

\newcommand{\divv}{\mbox{$\div \mkern-16mu /$\,\,}}

\def\f12{\frac 1 2}

\newcommand{\cala}{\mathcal{A}}
\newcommand{\si}{\Sigma}

\newcommand{\rrr}{\mathcal{R}}

\newcommand{\meg}{\geqslant}
\newcommand{\mik}{\leqslant}

\def\div{\text{div}}

\newcommand{\nabb}{\mbox{$\nabla \mkern-13mu /$\,}}
\newcommand{\lapp}{\mbox{$\triangle \mkern-13mu /$\,}}

\newtheorem{remark}{Remark}[section]

\newtheorem{theorem}{Theorem}[section]
\newtheorem{proposition}{Proposition}[subsection]

\begin{document}

\title[The trapping effect on degenerate horizons]{The trapping effect on degenerate horizons
\bigskip\\ \footnotesize{Y. Angelopoulos, S. Aretakis, D. Gajic}}

\date{December 29, 2015}

\begin{abstract}

We show that \textit{degenerate horizons} exhibit a new \textit{trapping effect}. Specifically, we obtain a non-degenerate Morawetz estimate for the wave equation in the domain of outer communications of extremal Reissner--Nordstr\"{o}m up to and including the future event horizon. We show that such an estimate requires 1) a higher degree of regularity for the initial data, reminiscent of the regularity loss in the high-frequency trapping estimates on the photon sphere, and 2) the vanishing of an explicit quantity that depends on the restriction of the initial data on the horizon.  The latter condition demonstrates that degenerate horizons exhibit a new $L^{2}$ concentration phenomenon (namely, a \textit{global} trapping effect, in the sense that this effect is not due to individual underlying null geodesics as in the case of the photon sphere). We moreover uncover a new \textit{stable} higher-order trapping effect; we show that higher-order estimates do not hold regardless of the degree of regularity and the support of the initial data. We connect our findings  to  the spectrum of the stability operator in the theory of marginally outer trapped surfaces (MOTS). Our methods and results play a crucial role in our upcoming works on linear and non-linear wave equations on extremal black hole backgrounds.

\end{abstract}

\maketitle

\tableofcontents

\section{Introduction}
\label{sec:Introduction}

\subsection{Overview}
\label{sec:Introduction1}

\textit{Black holes} are one of the most celebrated predictions of general relativity and as such their stability properties are of fundamental importance. The first step in resolving the non-linear stability problem for black holes is to establish quantitative decay estimates for the wave equation
\begin{equation}
\Box_{g}\psi=0
\label{we1}
\end{equation}
on fixed black hole backgrounds. One of the main difficulties in the analysis of the wave equation on such backgrounds is the so-called \textit{trapping effect} due to the existence of a family of \textit{trapped null geodesics} in the domain of outer communications whose limit point is the future timelike infinity. In the well-known \textit{Kerr--Newman} family of black holes, there is a family of  trapped null geodesics with constant Boyer-Lindquist coordinate $r$. In the special subfamily consisting of Schwarzschild backgrounds these trapped null geodesics span the hypersurface $r=3M$ known as the \textit{photon sphere}. Here $M$ is the mass parameter. From every other point in the Schwarzschild exterior region there is a codimension-one subset of future-directed null directions whose corresponding geodesics approach the photon sphere, and all other null geodesics either cross the event horizon $\mathcal{H}=\left\{r=r_{hor}\right\}$ or terminate at null infinity. 

The trapped null geodesics pose a well-known \textit{high-frequency obstruction}, known as \textit{the trapping effect},  for the existence of a \textit{non-degenerate  Morawetz estimate} for the wave equation of the form
\begin{equation}
\int_{0}^{\tau}\int_{\Sigma_{t}\cap\left\{r_{hor}<R_{1}\leq r\leq R_{2}\right\}}|\partial\psi|^{2} \, dg_{\Sigma_{t}}\, dt\leq \, C\int_{\Sigma_{0}}|\partial\psi|^{2}\, dg_{\Sigma_{0}},
\label{nottrue}
\end{equation}
where the trapped null geodesics exist in the region $\left\{r_{hor}<R_{1}\leq r\leq R_{2}\right\}$. The obstruction for \eqref{nottrue} originates from the existence of  high-frequency solutions to the wave equation with finite initial energy which are supported in a neighborhood of trapped null geodesics for arbitrarily long times.  This phenomenon has long been studied in the context of the obstacle problem for the wave equation in Minkowski space, where the analogue of trapped null geodesics are null lines which reflect off the obstacle's boundary. Recently, Sbierski \cite{janpaper}, building on previous work of Ralston \cite{ralston1} on the Gaussian beam approximation, showed that, on general Lorentzian manifolds, the energy at time $t$ of the localized high frequency solutions is comparable to the energy of the underlying null geodesic at time $t$. The fact that the energy of the trapped null geodesics in the Kerr--Newman family is constant immediately contradicts estimate \eqref{nottrue}. 
On the other hand, it can be shown (see, for instance, \cite{lecturesMD,part3,semyon2,jaredw} and references there-in) that on sub-extremal backgrounds estimate \eqref{nottrue} holds if the right hand side loses derivatives (i.e~includes higher order energies).

Having introduced the trapping effect on the photon sphere we next consider the event horizon. The event horizon is a null hypersurface ruled by null geodesics, known as the \textit{null generators}. In the \textit{sub-extremal} case, however, one  can show a local integrated decay Morawetz estimate in a neighborhood of the event horizon without any loss of differentiation (see \cite{lecturesMD}). This is possible because the energy of the null generators \textit{decays exponentially} in time $t$, in view of the so-called \textit{redshift effect} which in turn is based on the positivity of the \textit{surface gravity} (see Section \ref{background} for an introduction to these notions). Hence, no trapping takes place on the event horizon of sub-extremal black holes, even though the latter is ruled by null geodesics.

Nonethelesss, the situation is drastically different for \textit{degenerate horizons}, which are null hypersurfaces with vanishing surface gravity. This is the case for the event horizon of \textit{extremal} black holes. In this case the null generators have energy that is  constant in time. Hence, Sbierski's result implies that loss of regularity is a \textit{necessary} condition for the existence of a non-degenerate Morawetz estimate up to and including the event horizon.  However, until now there had not been any works providing \textit{sufficient} conditions. In fact, it was not known if generic solutions satisfy a Morawetz estimate.

   In a series of papers \cite{SA10, aretakis1,aretakis2} the second author initiated the mathematical study of the wave equation on extremal black holes and obtained a \textit{mixture of stability and instability results}. Specifically, it was shown that solutions to the wave equation decay in time towards the future, first-order derivatives remain bounded but \textit{do not decay} along the event horizon whereas higher order derivatives asymptotically \textit{blow up} along the event horizon. The derivatives are here taken with respect to the Eddington--Finkelstein coordinate system $(v,r,\theta,\phi)$; indeed, it is the $\partial_{r}\psi$ derivative of solutions $\phi$ to the wave equation that does not decay along the horizon. Note that this is a natural derivative to consider since in the extremal case $v$ is an affine parameter along null geodesics (this is in contrast to the situation in subextremal black holes, where the corresponding parameter is logarithmically related to the geodesic parameter). Hence, the previous instability results hold in Gaussian (geodesic) coordinates and are not merely an effect of imposing coordinates based on the Killing vector and observers at null infinity.

	We remark that these stability and instability results do \textbf{not} suffice to determine whether a non-degenerate Morawetz estimate up to and including the event horizon holds for general smooth solutions to the wave equation on extremal black holes (see  Section \ref{sec:PreviousWorks} for more details). It is important to emphasize that such a Morawetz estimate would play a crucial role in rigorously testing the validity of the aforementioned stability and instability features in non-linear settings. 

	\subsection{Summary of results and techniques}
	\label{sec:OverviewOfResultsAndTechniques}

In this paper we derive necessary and sufficient conditions for the existence of a non-degenerate Morawetz estimate in the domain of outer communications of extremal Reissner--Nordstr\"{o}m backgrounds up to and including the event horizon. Such estimates play a fundamental role in the analysis of non-linear wave equations and hence  necessary and sufficient conditions for their existence are relevant for the study of the black hole stability problem.

\medskip

\textbf{Summary of results}

\medskip

We obtain a complete characterization of the trapping effect on the event horizon of extremal Reissner--Nordstr\"{o}m. Specifically, we obtain a non-degenerate Morawetz estimate (see Theorem \ref{theo1} in Section \ref{sec:TheMainTheorems}) in the domain of outer communications of extremal Reissner--Nordstr\"{o}m \textit{up to and including the future event horizon}. We show that such an estimate requires
\begin{enumerate}
	\item \textit{a higher degree of regularity for the initial data,}
	\item \textit{the vanishing of the quantity $H[\psi]$ given by \eqref{defh}. We remark that $H[\psi]$ depends only on the restriction of the initial data on the horizon.}
\end{enumerate}
Note that the first condition is reminiscent of the regularity loss in the high-frequency trapping estimates on the photon sphere. In fact, we will show a result in the converse direction (see Theorem \ref{theo3} in Section \ref{sec:TheMainTheorems}):
\begin{itemize}
	\item \textit{if a weighted higher-order norm of the initial data is infinite then no non-degenerate Morawetz estimate holds.}
\end{itemize}
This implies that the loss of regularity of our result is \textit{optimal}. We also prove that the vanishing of the quantity $H[\psi]$ is \textit{necessary} in the following sense (see Theorem \ref{theo2} in Section \ref{sec:TheMainTheorems}):
 \begin{itemize}
	\item \textit{if the quantity $H[\psi]$ given by \eqref{defh} is initially non-zero then no non-degenerate Morawetz estimate holds, regardless of the degree of regularity of the initial data.}
\end{itemize}
The above result demonstrates that degenerate horizons exhibit a new \textit{global trapping effect}, in the sense that this effect is not due to individual underlying null geodesics as in the case of the photon sphere. This global trapping effect is closely related to the spectrum of the \textit{stability operator} for the sections of the event horizon (see Section \ref{sec:RelationWithTheoryOfMOTS}).

Furthermore, a new \textit{stable higher order trapping effect} is uncovered (see Theorem \ref{theo4} in Section \ref{sec:TheMainTheorems}), in the sense that 
\begin{itemize}
	\item \textit{higher order estimates up to and including the event horizon generically do \textbf{not} hold regardless of the degree of regularity and the support of the initial data.}
\end{itemize}

	\medskip
	
\textbf{	Summary of techniques}
	
	\medskip

We decompose $\psi$ as follows (see Section \ref{sec:PoincarEInequality})
\[\psi=\Big(\int_{\mathbb{S}^{2}}\psi\Big)+\Big(\psi-\int_{\mathbb{S}^{2}}\psi\Big) \]
 and for each componenet we use novel physical space\footnote{No decomposition in time frequencies is needed.} vector field multipliers. 
	
 For the spherical mean we use the singular vector field 
\[S=-\frac{1}{r-r_{hor}}\cdot Y\]
as a \textit{multiplier} vector field (see Section \ref{flux} for an introduction to the vector field method). Here $r_{hor}$ is the radius of the event horizon and $Y$ is a regular translation-invariant \textit{transversal} to the event horizon  vector field (see Section \ref{background}).  Clearly, $S$ is singular on the event horizon. The choice of this singular multiplier is motivated by the spectral properties of the stability operator for MOTS (see Section \ref{sec:RelationWithTheoryOfMOTS}). It seems that it is the first time that properties of the stability operators are used to yield estimates for hyperbolic equations. 

  We apply the divergence identity in regions $\mathcal{A}_{r_{0}}$ (see Section \ref{background}, figure \ref{figure2}) which do not include the event horizon if $r_{0}>r_{hor}$ and study the limiting behavior of the resulting equations are $r_{0}\rightarrow r_{hor}$. We uncover a special structure of the geometry of degenerate horizons to show that all the resulting singular terms can be estimated by singular norms of the initial data. The boundedness of these singular norms (which is an assumption on the initial data only) implies a non-degenerate Morawetz estimate. 

One of the most critical terms is the integral $\int_{\mathcal{A}_{r_{0}}}\frac{1}{r-r_{hor}}T\psi\cdot Y\psi$, where $T$ is the stationary Killing field (see \eqref{i4}). Clearly the integral of the integrand quantity $\frac{1}{r-r_{hor}}T\psi\cdot Y\psi$ over a spacial slice $\Sigma_{\tau}$ is infinite. However, we were able to show that if we first integrate $\frac{1}{r-r_{hor}}T\psi\cdot Y\psi$ over time and then over space then the resulting expression has a finite limit as $r\rightarrow r_{hor}$. In Section \ref{sec:Discussion} it is shown that in the sub-extremal case the corresponding limit is infinite demonstrating thus a new distinctive feature of degenerate horizons.  

For the component $\psi-\int_{\mathbb{S}^{2}}\psi$,  we use the vector field $Y$ as a \textit{commutator} vector field (see Section \ref{flux}) and the  \textit{degenerate} vector field $\bar{\partial}=-(r-r_{hor})\cdot Y$ as multiplier vector field. Our estimates rely on appropriate use of Hardy and Poincar\'{e} inequalities and a novel  structure of the wave equation in a neighborhood of the degenerate event horizon
which yields various cancellations of the most dangerous terms. 	It is worth contrasting this with the trapping estimate at the photon sphere where one is required to commute with either the Killing field $T$ or the standard Killing fields $\Omega_{i}, i=1,2,3$ of the sphere. All these vector fields are tangential to the photon sphere.

We note that our methods and results play a crucial role in our upcoming works on linear and non-linear wave equations on extremal black hole backgrounds.

\subsection{Previous results}
\label{sec:PreviousWorks}

To put our results into context, we briefly summarize previous work on the wave equation on black hole backgrounds. The study of the wave equation (\ref{we1}) on black hole backgrounds has a long history, starting in 1957 with the pioneering work of Regge--Wheeler \cite{regge} on the mode stability of (\ref{we1}) on Schwarzschild ($a=0$). The first quantitative result in Schwarzschild was obtained in 1987 by Kay--Wald \cite{Kay1987}, who proved uniform boundedness of solutions to the wave equation. In the last two decades there have been many (partial) results on the asymptotic behaviour of linear waves in the domain of outer communication of sub-extremal Kerr backgrounds, for which $|a|<M$, culminating in the proof of polynomial decay in time for solutions to (\ref{we1}) on the full sub-extremal range $|a|<M$ of Kerr backgrounds by Dafermos--Rodnianski--Shlapentokh-Rothman \cite{part3}; see also \cite{part3, lecturesMD,blukerr,tataru2} for a comprehensive list of references to earlier works. We also refer the reader to the inverse logarithmic decay results for solutions to the wave equation on spacetimes exhibiting \textit{stable} trapping  \cite{gusmu1,molog,keir}. 

The rigorous study of mathematical properties of the wave equation on extremal black holes was initiated in a series of papers \cite{aretakis1,aretakis2,aretakis3,aretakis4} where a mixture of stability and instability results was presented. Subsequent works of  Reall, Murata, Lucietti et al  \cite{harvey2013,murata2012,hm2012,hj2012} studied in a series of papers these instability properties on more general linear and non-linear settings.  The authors of \cite{harvey2013} numerically investigated spherically symmetric perturbations of extremal Reissner--Nordstr\"{o}m in the context of the Cauchy problem for the Einstein--Maxwell-scalar field system and discovered that, even for arbitrarilly
small initial perturbations, derivatives of the scalar field grow (in complete agreement with the work in the linear case). 
Ori \cite{ori2013} and Sela \cite{sela} have numerically investigated the relation of long time dynamics of scalar fields and the conservation laws.   Rigorous non-linear results have appeared in \cite{aretakis2013,yannis1}. Further applications and extensions have been presented in  \cite{gusmu1,iapwnes, dd2012,aretakis2012}. For work in the interior of extremal black holes we refer to \cite{gajic}.

\section{Background on the geometry of extremal R--N}
\label{background}

\subsection{The extremal R--N spacetime}
\label{sec:TheExtremalRNSpacetime}

We define the extremal Reissner--Nordstr\"om spacetime as the Lorentzian manifold $(\mathcal{M},g)$, where $\mathcal{M}=\mathbb{R}\times \mathbb{R}_{+}\times \mathbb{S}^2$. We introduce the \emph{ingoing Eddington--Finkelstein coordinate chart} $(v,r,\theta,\varphi)$, where $v\in \mathbb{R}$, $r\in \mathbb{R}_+$ and $(\theta,\varphi)$ are the standard polar coordinates on the round sphere $\mathbb{S}^2$. In these coordinates, the metric can be expressed as:
\begin{equation*}
g=-Ddv^2+2dvdr+r^2(d\theta^2+\sin^2\theta d\varphi^2),
\end{equation*}
where $D:= r^{-2}(M-r)^2$. This metric, along with the associated Faraday tensor
\[F=M/r^2 dr \wedge dv+ M \sin \theta d\theta \wedge d\phi,\]  is a solution to the Einstein--Maxwell equations. 

The null hypersurface $\mathcal{H}:=\{r=M\}$ is called the \emph{future event horizon}. The region \\$\mathcal{M}_{\rm ext}\doteq \mathcal{M}\cap \{r>M\}$ is known as the \emph{exterior region} or \emph{domain of outer communications} of extremal Reissner--Nordstr\"om. The \emph{interior region} $\mathcal{M}_{\rm int} \doteq \mathcal{M}\cap \{0<r<M\}$ will not feature in the remainder of this paper.

We will refer to the coordinate $v$ as the \emph{advanced null coordinate}. We can introduce a \emph{retarded null coordinate} $u$, given by $u=v-2r_*$, where $r_*$ is defined as a solution to $\frac{dr_*}{dr}=D^{-1}$ and is given explicitly by
\begin{equation*}
r_*(r)=\frac{M^2}{M-r}+2M\log(M-r)+r.
\end{equation*}

The tuple $(u,v,\theta,\varphi)$ constitutes a coordinate chart on $\mathcal{M}_{\rm ext}$, commonly referred to as \emph{Eddington--Finkelstein double-null coordinates}. In these coordinates, the metric is given by
\begin{equation*}
g=-Ddudv+r^2(d\theta^2+\sin^2\theta d\varphi^2).
\end{equation*}

\begin{figure}[h!]
\begin{center}
\includegraphics[width=2.in]{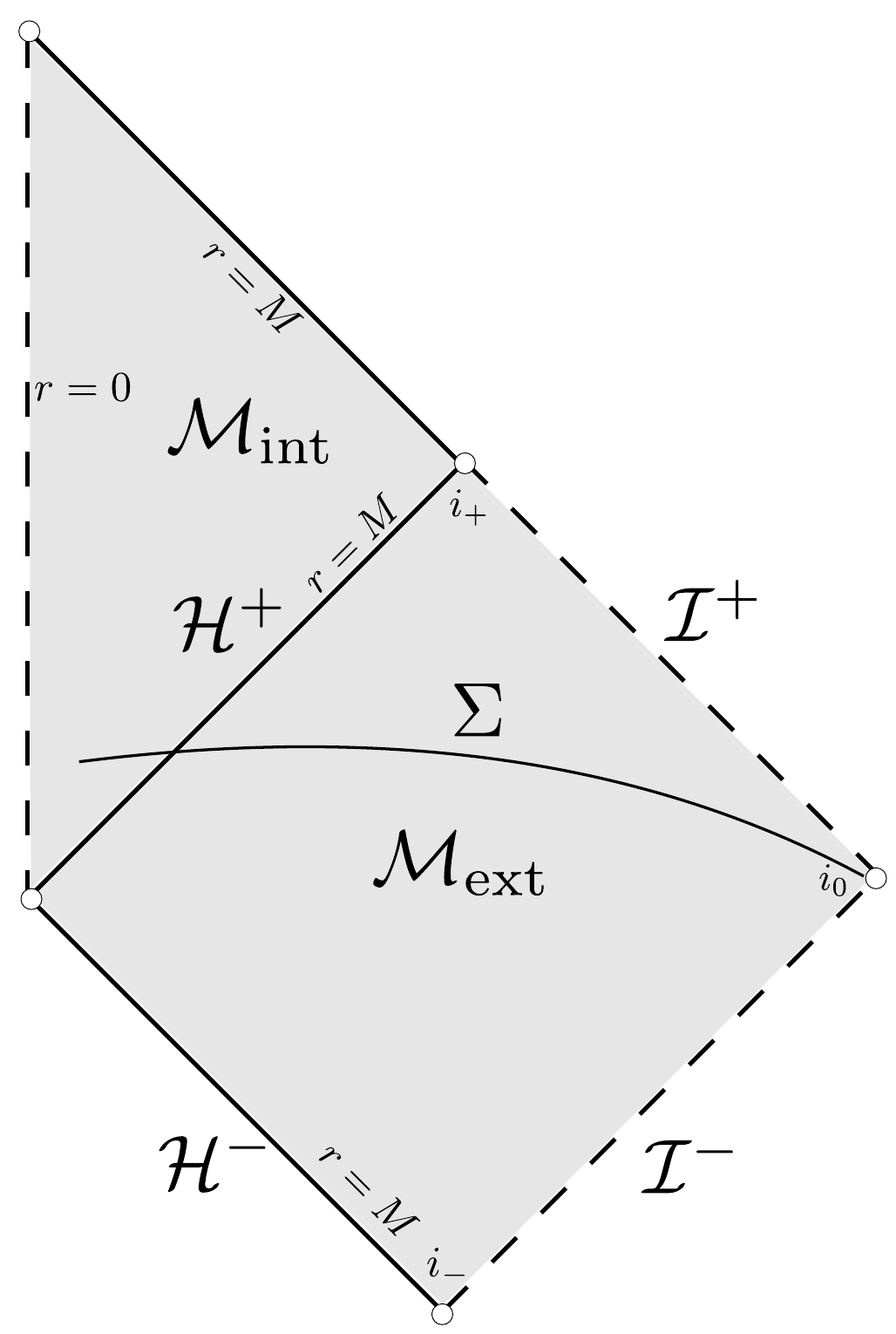}
\caption{\label{fig:PDextremalRN} The Penrose diagram of extremal Reissner--Nordstr\"om}
\end{center}
\end{figure}

Let $\Sigma$ be a spherically symmetric asymptotically flat spacelike hypersurface in $\mathcal{M}$ such that $\Sigma\cap \mathcal{H}\neq \emptyset$. Let $M<R<R_2$. We define $\Sigma_0$ as:
\begin{equation*}
\Sigma_0\doteq (\Sigma \cap \{R\leq r\leq R_2\})\cup \{v=v_{\Sigma}(R),\,M\leq r\leq R\}\cup \{u=u_{\Sigma}(R_2),\,v_{\Sigma}(R_2)\leq v<\infty\}.
\end{equation*}
Here, $v_{\Sigma}$ and $u_{\Sigma}$ denote the restrictions to $\Sigma$ of the functions $v$ and $u$ to, respectively. See Figure \ref{figure2}. In this paper, we will only consider the region $J^+(\Sigma_0)$, which is foliated by the hypersurfaces $\Sigma_{\tau}=\varphi^{\partial_{v}}_{\tau}(\Sigma_{0})$, where $\varphi^{\partial_{v}}_{\tau}$ is the flow generated by the vector field $\partial_{v}$, the latter being taken with respect to the ingoing E--F coordinates. We  next define several regions and hypersurfaces that will be very important in our analysis. 
\begin{figure}[h!]
\begin{center}
\includegraphics[width=3in]{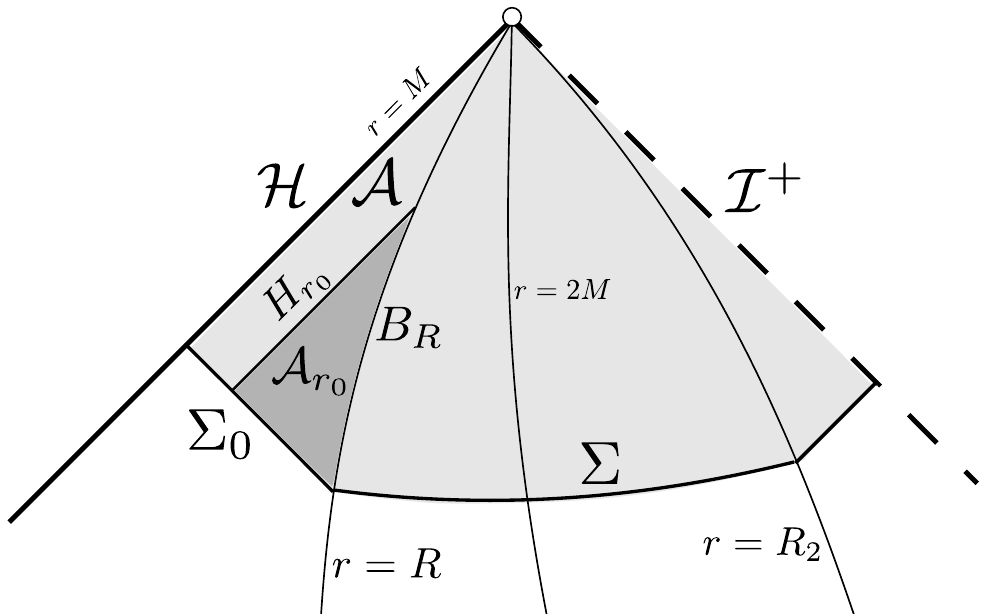}
\caption{Representation of spacetime regions}
\end{center}
\label{figure2}
\end{figure}
Consider the outgoing hypersurfaces $H_{r_0}$ which are defined as follows:
\begin{equation*}
H_{r_0}\doteq J^+(\Sigma_0)\cap\{u=u_{\Sigma_0}(r_0),\,r\leq R\}
\end{equation*}
and the following timelike hypersurface:
\begin{equation*}
B_R\doteq \{r=R\}\cap J^+(\Sigma_0).
\end{equation*}
We will frequently restrict to the region $\mathcal{A}_{r_0} \subset J^+(\Sigma_0)$, where
\begin{equation*}
\mathcal{A}_{r_0}\doteq J^+(\Sigma_0)\cap \left\{r\leq R \right\}\cap J^-(H_{r_0})
\end{equation*}
and we vary $r_0\in (M,R)$. Here, $u_{\Sigma_0}$ denotes the function $u$ restricted to $\Sigma_0$. Note that $H_{r_0}$ approaches $\mathcal{H}$ in the limit $r_0\downarrow M$. We define the corresponding limiting region $\mathcal{A}$ as follows:
\begin{equation*}
\mathcal{A}\doteq J^+(\Sigma_0)\cap \{M\leq r \leq R\}.
\end{equation*}
Let, finally, $\textbf{n}_{\Sigma}$ denote the (future) normal to the hypersurface $\Sigma$.

\subsection{Photon sphere}
\label{sec:geom_photon}
By the causal geometry of $\mathcal{M}_{\rm ext}$, one can easily infer the existence of geodesics that do not cross $\mathcal{M}$ or approach future null infinity $\mathcal{I}^+$. A class of these geodesics $\gamma:\rr \to \mathcal{M}_{\rm ext}$ can be parametrized as follows:
\begin{equation*}
\gamma(\tau)=(v(\tau),2M,\frac{\pi}{2},\varphi(\tau)),
\end{equation*}
where $v(\tau)$ and $\varphi(\tau)$ depend linearly on $\tau$. The timelike hypersurface $\{r=2M\}$ is called the \emph{photon sphere}. As was mentioned before, the existence of these geodesics gives rise to the trapping effect of the photon sphere. This is manifested in the loss of regularity to Morawetz estimates for the wave equation \eqref{we1} (see estimate \ref{dmora}).

\subsection{The red-shift effect}
\label{sec:geom_redshift}
The vector field $T\doteq \partial_v$ is a causal Killing vector field that is timelike in $\mathcal{M}_{\rm ext}$ and null on $\mathcal{H}$. Correspondingly, we can define the \emph{surface gravity of $\mathcal{H}$ with respect to $T$} as the function $\kappa: \mathcal{H}\to \mathbb{R}$, such that
\begin{equation*}
\nabla_T T|_{\mathcal{H}}= \kappa \cdot T|_{\mathcal{H}}.
\end{equation*}
In sub-extremal Reissner--Nordstr\"om spacetimes, the surface gravity, which is analogously defined, is strictly positive. However, in extremal Reissner--Nordstr\"om we have that $\kappa=0$. A horizon with a corresponding vanishing surface gravity is called a \emph{degenerate horizon}.

The relevance of the sign of $\kappa$ can be understood geometrically by considering the null geodesics $\gamma$ that generate $\mathcal{H}$. The energy of null geodesics in Reissner--Nordstr\"om with respect to the timelike $\partial_{v}-$invariant vector field $\partial_{v}-\partial_{r}$ decays exponentially in time with an exponent proportional to $\kappa$. The corresponding physical interpretation is that the frequency of a light signal sent by an observer crossing the event horizon to another observer crossing the event horizon at a later time is red-shifted. This \textit{redshift effect} in the context of decay results for the wave equation was first used by Dafermos and Rodnianski in \cite{redshift}. The degenaracy of the redshift effect in the extremal case introduces an additional difficulty in the analysis of the wave equation on such backgrounds.

\subsection{Energy currents and the vector field method}
\label{flux}
The energy-momentum tensor $\mathbf{T}$ corresponding to the wave equation (\ref{we1}) is a symmetric 2-tensor with components
\begin{equation*}
\mathbf{T}_{\alpha \beta}[f]=\partial_{\alpha}f\partial_{\beta}f-\frac{1}{2}g_{\alpha\beta}(g^{-1})^{\gamma \delta}\partial_{\gamma}f\partial_{\delta}f.
\end{equation*}
Moreover,
\begin{equation*}
\nabla^{\beta}\mathbf{T}_{\alpha \beta}[f]=\partial_{\alpha}f\square_gf,
\end{equation*}
so $\mathbf{T}[\psi]$ is divergence-free for solutions $\psi$ to (\ref{we1}).

Let $V$ be a vector field. Then we denote the energy current with respect to $V$ by $J^V$, where $J^V=\mathbf{T}(V,\cdot)$. Let $W$ be another vector field, then we denote
\begin{equation*}
J^V[f]\cdot W=\mathbf{T}[f](V,W).
\end{equation*}
An immediate calculation yields
\begin{equation*}
\textnormal{div}\,J^V[f]=K^V[f]+\mathcal{E}^V[f],
\end{equation*}
where
\begin{align*}
K^V[f]\doteq&\: \mathbf{T}^{\alpha\beta}\nabla_{\alpha}V_{\beta},\\
\mathcal{E}^V[f]\doteq&\: V^{\alpha}\nabla^{\beta}\mathbf{T}_{\alpha \beta}[f]=V(f)\square_gf.
\end{align*}
Note that $\mathcal{E}^V[\psi]=0$ for solutions $\psi$ to (\ref{we1}). Furthermore, $K^V[f]=0$ if $V$ is a Killing vector field. 
The \emph{vector field method} (see also \cite{muchT}) comprises a careful choice of the vector field $V$ (the \emph{vector field multiplier}) and $W$, where $f=W\psi$ (the \emph{commutation vector field}), so as to obtain suitable energy estimates after applying Stokes' theorem in appropriate spacetime regions. 

Consider the vector fields
\begin{equation}
T:= \partial_v, \ \ \ Y:= \partial_r,
\label{ty}
\end{equation}
with respect to the $(v,r,\theta,\phi)$ coordinate system.  The vector field $T$ is timelike in $\mathcal{M}_{\rm ext}$ and null along $\mathcal{H}$. The vector field $Y$ is regular and transversal to $\mathcal{H}$. 

We introduce the vector fields $P,N$ (see also \cite{aretakis1,aretakis2}), which \textit{close} to the event horizon satisfy the following
\begin{equation*}
\begin{split}
& P\sim T-(r-M)\cdot Y,\\
& N \sim T-Y,
\end{split}
\end{equation*}
 Observe that $N$ is future-directed and \emph{strictly} timelike in $J^+(\Sigma_0)$. On the other hand, the vector field $P$ is timelike in $\mathcal{M}_{\rm ext}$ and null at $\mathcal{H}^+$.

One can easily check that close to the event horizon it holds
\begin{equation}
 \begin{split}
J^T[f]\cdot N&\sim (Tf)^2+ (r-M)^{2}\cdot(Yf)^2+|\nabb f|^2,\\
J^P[f]\cdot N&\sim (Tf)^2+ (r-M)\cdot(Yf)^2+|\nabb f|^2,\\
J^N[f]\cdot N&\sim (Tf)^2+ (Yf)^2+|\nabb f|^2,
\end{split}
\label{fluxestpn}
\end{equation}
where $\nabb$ denotes the covariant derivative restricted to the spheres of constant $v$ and $r$.

\section{The main theorems}
\label{sec:TheMainTheorems}

For the definition of relevant notions and notation and, in particular, for the definition of specific regions, hypersurfaces, vector fields and their fluxes, see section \ref{background}.

We study the Cauchy problem for the linear wave equation \eqref{we1} on the exterior $\mathcal{M}$ of the extreme Reissner--Nordstr\"{o}m spacetime up to and including the event horizon $\mathcal{H}$ with initial data
\[\left.\psi\right|_{\Sigma_{0}}\in H^{s}_{\text{loc}}(\Sigma_{0}), \ \left. D_{\text{tran}}\psi\right|_{\Sigma_{0}}\in H^{s-1}_{\text{loc}}(\Sigma_{0}), \]
where $D_{\text{tran}}$ is a transversal to $\Sigma_{0}$ vector field. The hypersurface $\Sigma_{0}$ is as defined in Section \ref{background}. We assume that $s$ is sufficiently large so that all the weighted norms of our estimates are finite. 

\subsection{Summary of previous results}
\label{sec:SummaryOfPreviousResults}

For the reader's convenience we recall some of the key results of \cite{aretakis1, aretakis2} for solutions $\psi$ to \eqref{we1} on extremal Reissner--Nordstr\"{o}m:

\bigskip

\textbf{(1) Energy boundedness.} 
\begin{equation}\label{nfluxold}
\int_{\si_{\tau}} \left( J^V [\psi ] \cdot \textbf{n}_{\si_{\tau}} \right) dg_{\si_{\tau}} 
 \mik C \int_{\si_{0}} \left( J^V [\psi ] \cdot \textbf{n}_{\si_{0}} \right) dg_{\si_{0}} , 
\end{equation} 
where $V\in \left\{T,P,N\right\}$.

\bigskip

\textbf{(2) Integrated local energy decay.}

\begin{equation}\label{dmora}
\begin{split}
&\int_{\left\{M\leq r \leq r_{0}\right\}}  \Big( (T \psi )^2+ | \slashed{\nabla} \psi |^2  + (r-M)\cdot (Y \psi )^2  \Big) 
\\ &+ \int_{\left\{r_{0}\leq r\leq r_{1}\right\}}\Big( (\partial_{r^{*}}\psi)^{2}+(r-2M)^{2}\cdot\big((T \psi )^2  + | \slashed{\nabla} \psi |^2\big) \Big) 
 \\ \ \ & \mik C \int_{\si_0 } \left(  J^N [\psi ] \cdot \textbf{n}_{\si_0} \right) dg_{\si_0} ,
\end{split}  
\end{equation}
where $M<r_{0}<2M<r_{1}$. Note that this estimate degenerates both on the photon sphere and on the horizon. Specifically, the degeneracy applies to 
\begin{itemize}
	\item the \textit{tangential} derivatives to the photon sphere,
	\item the \textit{transversal} derivatives to the event horizon.
\end{itemize}
Removing the degeneracy at the photon sphere comes at the expense commuting with the Killing field $T$. Theorem \ref{theo1} below provides an estimate which does not degenerate at the horizon. 

\bigskip

\textbf{ (3) Energy decay }

\begin{equation}\label{tfluxdecayold}
\int_{\si_{\tau} } \left( J^T [\psi ] \cdot \textbf{n}_{\si_{\tau}} \right) dg_{\si_{\tau}} \mik \frac{C}{\tau^2}\cdot I[\psi] ,
\end{equation} 
where $I[\psi]$ is a suitable weighted norm of the initial data.

\bigskip

\textbf{ (4) Conserved quantities and instabilities on the horizon.} 

a) The quantity
\begin{equation}
\int_{\Sigma_{\tau}\cap\mathcal{H}}\left(Y\psi+\frac{1}{M}\psi\right)
\label{eqdfsfsd}
\end{equation}
is conserved along the event horizon (i.e.~independent of $\tau$). In fact, in \cite{aretakis2} a hierarchy of conservation laws was established for each angular frequency.

b) For generic solutions we have the following blow-up result
$$ \max_{\Sigma_{\tau}\cap\mathcal{H}}|Y^k T^m \psi | \rightarrow \infty , $$
asymptotically along $\mathcal{H}$ as $\tau\rightarrow\infty$, for $k \meg m+2$.

\subsection{The statements of the main theorems}
\label{sec:TheStatementsOfTheMainTheorems}
The main results of the present paper concern non-degenerate Morawetz estimates up to and including the event horizon $\mathcal{H}$. 

\begin{theorem}[\textbf{The trapping estimate}]

Consider the following weighted norm on the Cauchy hypersurface $\Sigma_{0}$
\begin{equation}
\begin{split}
D^{w}_{\Sigma_{0}}[\psi]: =& \ \int_{\Sigma_{0}\cap\left\{r\leq R\right\}}\left[\frac{1}{(r-M)}\cdot \left(Y\psi+\frac{1}{r}\psi\right)^{2}\right]dr\,d\omega\\&+\int_{\Sigma_{0}}\Big(J^{P}[\psi]\cdot \textbf{n}_{\Sigma_{0}}\Big)\, dg_{\Sigma_{0}}+\int_{\Sigma_{0}}\Big(J^{P}[T\psi]\cdot \textbf{n}_{\Sigma_{0}}\Big)\, dg_{\Sigma_{0}}\\& +\int_{\Sigma_{0}\cap\left\{r\leq R\right\}}\Big(J^{P}[Y\psi]\cdot \textbf{n}_{\Sigma_{0}}\Big)\, dg_{\Sigma_{0}},
\end{split}
\label{norminitial1}
\end{equation}
where the $J^{P}$ flux and the constant $R$ are as defined in Section \ref{background}.

 Then, there is a constant $C>0$ that depends only on the mass parameter $M$ such that for all solutions $\psi$ to the wave equation on extremal Reissner--Nordstr\"{o}m which arise from initial data with bounded $D^{w}_{\Sigma_{0}}$ norm the following non-degenerate Morawetz estimate holds
\begin{equation}
\begin{split}
\int_{\mathcal{A}}\Big[\psi^{2}+(Y\psi)^{2}+(T\psi)^{2}+\left|\nabb\psi\right|^{2}\Big]\, dg_{\mathcal{A}}\leq  & \ C \cdot D^{w}_{\Sigma_{0}}[\psi].  
\end{split}
\label{e1theo1}
\end{equation}
Here $\mathcal{A}$ is the spacetime region $\left\{M\leq r\leq R\right\}$ which, in particular, contains the future event horizon $\mathcal{H}=\left\{r=M\right\}$. 
\label{theo1}
\end{theorem}

\begin{remark}
Observe that the right hand side of \eqref{e1theo1} requires higher regularity than its left hand side. Furthermore, the boundedness of the $D^{w}_{\Sigma_{0}}$ norm forces the conserved charge
\[H[\psi]=\int_{\Sigma_{0}\cap\mathcal{H}}\left(Y\psi+\frac{1}{M}\psi\right)d\omega\]
to vanish. Note that if the data are in $C^{2}$ and the conserved charge $H[\psi]$ vanishes, then the integral 
\[\int_{\Sigma_{0}\cap\left\{r\leq R\right\}}\left[\frac{1}{(r-M)}\cdot \left(Y\psi+\frac{1}{r}\psi\right)^{2}\right]dr\,d\omega \]
is bounded  by the remaining three integrals in \eqref{norminitial1}. 

Hence not only Theorem \ref{theo1} requires high regularity for the initial data but also requires the vanishing of the conserved charge $H[\psi]$. The next two theorems show that this result is optimal by providing results in the converse direction.
\label{remarktheo1}
\end{remark}

\begin{theorem}[\textbf{Trapping and conserved charges}]
If $\psi$ is a solution to the wave equation on extremal Reissner--Nordstr\"{o}m with smooth compactly supported initial data for which the conserved charge
\begin{equation}
H[\psi]:=\int_{\Sigma_{0}\cap\mathcal{H}}\left(Y\psi+\frac{1}{M}\psi\right)d\omega\neq 0
\label{defh}
\end{equation}
then
\[\int_{\mathcal{A}}\left(Y\psi\right)^{2}\, dg_{\mathcal{A}}=\infty. \]
Hence, no non-degenerate Morawetz estimate holds for such solutions $\psi$ to the wave equation.
\label{theo2}
\end{theorem}
\begin{remark}
Theorem \ref{theo2} shows that trapping takes place on degenerate horizons even for smooth solutions as long as their conserved charge is non-vanishing.  It, in particular, implies that the trapping effect on degenerate horizon is not only due to a high frequency obstruction that requires loss of regularity but also due to global properties of the horizon which are independent of the degree of regularity of the initial data. This yields a new $L^{2}$-concentration phenomenon for degenerate horizon. Hence, the degenerate horizon $\mathcal{H}$ should be thought of as being trapped as a whole. This is in stark contrast with the trapping effect at the photon sphere where trapping is due to high frequency solutions which are supported for arbitrarily long times in small neighborhoods of individual null geodesics on the photon sphere. 
\label{remarktheo2}
\end{remark}

The next theorem shows that the regularity required in Theorem \ref{theo1} for the non-degenerate estimate \eqref{e1theo1} to hold is in fact \textit{optimal}. Specifically, we will show that 
 the non-degenerate spacetime integral $\int_{\mathcal{A}}|\partial\psi|^{2}$ is not bounded if we assume that the initial data are less regular than required for the boundedness of $D^{w}_{\Sigma_{0}}$. We have the following

\begin{theorem}[\textbf{Trapping and optimal loss of regularity}]
Consider initial data for the wave equation on extremal Reissner--Nordstr\"{o}m such that the conserved charge is vanishing
\begin{equation}
H[\psi]=0
\label{optimalre1}
\end{equation} 
and
\begin{equation}
\int_{\Sigma_{0}\cap\left\{r\leq R\right\}}\left[\frac{1}{(r-M)}\cdot \left(Y\psi+\frac{1}{r}\psi\right)^{2}\right]dr\,d\omega
\label{optimalre2}=\infty.
\end{equation}
Clearly condition \eqref{optimalre2}  implies that 
\[D_{\Sigma_{0}}^{w}[\psi]=\infty. \]
Then, no non-degenerate Morawetz estimate holds for $\psi$, that is 
\[\int_{\mathcal{A}}\left(Y\psi\right)^{2}\, dg_{\mathcal{A}}=\infty. \]
\label{theo3}
\end{theorem}
The last theorem concerns higher order trapping estimates. Specifically, we show that \textit{higher order stable trapping takes place on degenerate horizons}. We have the following
\begin{theorem}[\textbf{Stable higher-order trapping}]
Generic solutions to the wave equation on extremal Reissner--Nordstr\"{o}m with generic smooth initial data which are supported in $\left\{M<R_{1}\leq r\leq R_{2}\right\}$ satisfy
\[ \int_{\mathcal{A}}\big(Y^{k}\psi\big)^{2}\, dg_{\mathcal{A}}=\infty, \text{ for }k\geq 2\]
and hence no non-degenerate higher order Morawetz estimate holds. 
\label{theo4}
\end{theorem}
 Hence, no non-degenerate higher order Morawetz estimate holds  even for initial data which are compactly supported and supported away from the horizon and as such any charge on the event horizon initially vanishes and any weighted higher order norm is finite. This implies that stable higher order trapping takes place on degenerate horizons. 

\begin{remark} As will become evident from the proof, the divergence of the above integrals arises from spherically symmetric solutions to the wave equation. Given that the spacetime metric is spherically symmetric, one can decompose a smooth solution $\psi$ to the wave equation to solutions $\psi_{\ell}$ supported on individual spherical harmonic parameters $\ell\geq 0$ (see Section \ref{sec:PoincarEInequality} below). Then, using the techniques developed in this paper, one expects to be able to show \textbf{blow-up for higher spherical harmonic parameters}:
If $\psi_{\geq \ell}$ is a smooth solution to the wave equation supported on angular frequencies greater or equal to $\ell$ such that the higher-order conserved charge (defined in \cite{aretakis2})  satisfies the generic condition $H_{\ell}[\psi_{\geq\ell}]\neq 0$ then
\[\int_{\mathcal{A}}\left(Y^{k}\psi_{\geq \ell}\right)^{2}=\infty, \]
for $k\geq \ell+1$. This extended result would be particularly useful for $\ell\geq 2$, since as is well known there are no dynamical solutions of the linearized Einstein equations with spherical harmonic parameter $\ell=0,1$ (see \cite{newdaf}).
\label{remarkl}
\end{remark}

\section{The non-degenerate Morawetz estimate}
\label{sec:NonDegenerateMorawetzEstimate}

\subsection{Elliptic theory and Poincar\'{e}'s inequality}
\label{sec:PoincarEInequality}
We recall briefly some basic facts of the spectral theory of a standard sphere $\mathbb{S}^2 (r)$ of radius $r > 0$. The space $L^2 (\mathbb{S}^2 (r) )$ admits the following orthogonal decomposition:
$$  L^2 (\mathbb{S}^2 (r) ) = \oplus_{l=0}^{\infty} E_l , $$
where the eigenspaces $E_l$ are of dimension $2l+1$, and their corresponding eigenfunctions are denoted by $Y^{m,\ell}$ for $m \in \mathbf{Z} \cap [-\ell , \ell ]$ (the functions $Y^{m,\ell}$ are usually referred to as \textit{spherical harmonics}). The eigenvalues of the spherical Laplacian $\slashed{\Delta}$ are equal to $-\frac{\ell (\ell+1)}{r^2}$.

 Hence, any function $f \in L^{2} \left( \mathbb{S}^2 (r)\right)$ can be written as:
$$ f = \sum_{l=0}^{\infty} \sum_{m=-\ell}^\ell f_{m,\ell} (r) \cdot Y^{m,\ell}  \doteq \sum_{\ell =0}^{\infty} f_\ell  \doteq \left( \sum_{\ell=0}^{K-1} f_\ell \right) + f_{ \meg K} , $$
for any $K \meg 1$, where we denote by $f_{\ell}$ the projection of $f$ on the eigenspace $E_{\ell}$.

In view of the spherical symmetry of the extremal Reissner--Nordstr\"{o}m spacetime $\mathcal{M}$, if decompose any solution $\psi$ of the linear wave equation on $\mathcal{M}$:
\begin{equation}
 \psi =  \sum_{\ell=0}^{\infty} \psi_\ell ,
\label{sphedeco}
\end{equation}
then every projection $\psi_{\ell}$ will also satisfy the wave equation.  For example, $\psi$ can be uniquelly written as 
\begin{equation}
\psi=\psi_{0}+\psi_{\geq 1},
\label{spdeco}
\end{equation}
where $\psi_{0}=\int_{\mathbb{S}^{2}}\psi\, d\omega$ is the spherical mean of $\psi$. 

From now on, we will say that $\psi$ is supported on angular frequencies $ \ell \meg K$ for some $K \meg 0$ if initially we have that $\psi_k = 0$ for $k \in \mathbf{N} \cap [0, K-1]$, and that $\psi$ is supported on angular frequency $K$ if $\psi \in E_K$.

Finally we record here Poincar\'{e}'s inequality

\begin{proposition}
Let $f \in  L^{2}\left(\mathbb{S}^2 (r)\right)$ for some $r > 0$, and let $f_{\ell} = 0$ for $\ell \in \mathbf{N} \cap [0 , K-1 ]$ for some $K \in \mathbf{N}$, $K \meg 1$. Then we have that:
\begin{equation}\label{poincare}
\int_{\mathbb{S}^2 (r)} f^2 d\omega \mik \frac{r^2}{K (K+1)} \int_{\mathbb{S}^2 (r)} | \slashed{\nabla} f |^2 d\omega .
\end{equation}
Additionally we note that equality in \eqref{poincare} holds if  $f_{\ell} = 0$ for all $\ell \neq K$.
\label{yanprop}
\end{proposition}

\subsection{Hardy inequalities}
\label{sec:HardyInequalities}
We here list a few Hardy-type inequalities for functions defined on the exterior $\mathcal{M}$ of the extremal Reissner--Nordstr\"{o}m spacetime. 

The regions $\mathcal{A},\mathcal{A}_{r_{0}}$ and the hypersurfaces $B_{R},H_{r_{0}}$ are as defined in Section \ref{background}.
\begin{proposition}[First Hardy inequality]
Let $f:\mathcal{M}\to \mathbb{R}$ be a $C^1$ function. Let $p\in \mathbb{R}\setminus \{-1\}$ and suppose that the pointwise limit $\lim_{r\to M}(r-M)^{p+1}f^2(v,r,\theta,\phi)=0$. Then
\begin{equation}
\begin{split}
\label{1hardy}
\int_{\mathcal{A}}(r-M)^{p}f^2\,d\omega dr dv\leq &\  \frac{4}{(p+1)^2}\int_{\mathcal{A}}(r-M)^{p+2}(\partial_rf)^2\,d\omega dr dv\\ &+\frac{2}{p+1}\int_{B_R}(r-M)^{p+1}f^2\,d\omega dv.
\end{split}
\end{equation}
In particular, if $p<-1$, then we have that
\begin{equation}
\label{3hardy}
\int_{\mathcal{A}}(r-M)^{p}f^2\,d\omega dr dv\leq \frac{4}{(p+1)^2}\int_{\mathcal{A}}(r-M)^{p+2}(\partial_rf)^2\,d\omega dr dv.
\end{equation}
\end{proposition}
\begin{proof}
Integrate first $\partial_r((r-M)^{p+1}f^2)$ along constant-$v,\theta,\phi$ curves in $\mathcal{A}$ with $p\neq -1$ and use that $\lim_{r\to M}(r-M)^{p+1}f^2(v,r,\theta,\phi)=0$ to obtain:
\begin{equation}
\label{eq:hardyeq1}
\int_{M\leq r \leq r_{0}}(p+1)(r-M)^{p}f^2+2(r-M)^{p+1} f\partial_r f\,dr=(p+1)^{-1}\int_{r=R}(r-M)^{p+1}f^2.
\end{equation}
Clearly the constants are independent of $v,\theta,\phi$ and hence we can integrate in $v,\theta,\phi$ to obtain:
\begin{equation}
\label{eq:hardyeq}
\int_{\mathcal{A}}(p+1)(r-M)^{p}f^2+2(r-M)^{p+1} f\partial_r f\,d\omega dr dv=(p+1)^{-1}\int_{B_R}(r-M)^{p+1}f^2\,d\omega dv.
\end{equation}
We rearrange the terms above and multiply both sides by $(p+1)^{-1}$:
\begin{equation}
\begin{split}
\label{1hardyrea}
\int_{\mathcal{A}}(r-M)^{p}f^2\,d\omega dr dv=&\ (p+1)^{-1}\int_{B_R}(r-M)^{p+1}f^2\,d\omega dv\\& -2(p+1)^{-1}\int_{\mathcal{A}}(r-M)^{p+1}f\partial_r f\,d\omega dr dv.
\end{split}
\end{equation}
We apply a weighted Cauchy--Schwarz inequality to estimate
\begin{equation*}
\begin{split}
2(p+1)^{-1}&\int_{\mathcal{A}}(r-M)^{p+1}|f||\partial_r f|\,d\omega dr dv\\ &\leq\: \alpha \int_{\mathcal{A}}(r-M)^p f^2\,d\omega dr dv+\alpha^{-1}(p+1)^{-2} \int_{\mathcal{A}}(r-M)^{p+2} (\partial_rf)^2\,d\omega dr dv,
\end{split}
\end{equation*}
where $0<\alpha<1$. We use the above inequality together with (\ref{eq:hardyeq}) to obtain:
\begin{equation*}
\begin{split}
\int_{\mathcal{A}}(r-M)^{p}f^2\,d\omega dr dv\leq&\: \alpha^{-1}(1-\alpha)^{-1}(p+1)^{-2}\int_{\mathcal{A}}(r-M)^{p+2}(\partial_rf)^2\,d\omega dr dv\\
&+(1-\alpha)^{-1}(p+1)^{-1}\int_{B_R}(r-M)^{p+1}f^2\,d\omega dv.
\end{split}
\end{equation*}
The function $\alpha^{-1}(1-\alpha)^{-1}$ attains its minimum at $\alpha=\frac{1}{2}$. By taking $\alpha=\frac{1}{2}$ in the above inequality, we arrive at (\ref{1hardy}).
\end{proof}

\begin{proposition}[Second Hardy inequality]
Let $f:\mathcal{M}\to \mathbb{R}$ be a $C^1$ function. Let $r_1>r_0$. Then, for any $\epsilon>0$ we can estimate
\begin{equation}
\label{2hardy}
\begin{split}
&\int_{H_{r_0}}f^2\,d\omega dv\leq \\& \epsilon \int_{\mathcal{A}_{r_0}}(\partial_rf)^2\,d\omega dr dv+(1+\epsilon^{-1})\int_{\mathcal{A}_{r_0}}f^2\,d\omega dr dv+\int_{B_R}f^2\,d\omega dv.
\end{split}
\end{equation}
\end{proposition}
\begin{proof}
In a similar fashion as above, we integrate $\partial_r(f^2)$ over $\mathcal{A}_{r_0}$ to obtain:
\begin{equation*}
\begin{split}
\int_{H_{r_0}}f^2\,d\omega dv=&\:-2\int_{\mathcal{A}_{r_0}}f\partial_r f\,d\omega dr dv+\int_{B_R}f^2\,d\omega dv.
\end{split}
\end{equation*}
The inequality (\ref{2hardy}) follows immediately after applying a weighted Cauchy--Schwarz inequality on the first term.
\end{proof}

We finally state yet another Hardy inequality, the proof of which can be found in \cite{SA10} (Proposition 6.0.4):
\begin{proposition}[Third Hardy inequality]
Let $f:\mathcal{M}\to \mathbb{R}$ be a $C^1$ function. Let $r_1>r_0$ such that $M<r_{0}<r_{1}<2M$. We define the regions
\[\mathcal{A}=\mathcal{R}(0,\tau)\cap \left\{ M\leq r\leq r_{0} \right\}, \]
\[\mathcal{B}=\mathcal{R}(0,\tau)\cap \left\{r_{0}\leq r\leq r_{1} \right\}. \]
 Then, 
\begin{equation}
\label{3hardy}
\begin{split}
\int_{\mathcal{A}}f^{2}\leq C\int_{\mathcal{B}}f^{2}+C\int_{\mathcal{A}\cup\mathcal{B}}D\cdot [(\partial_{v}f)^{2}+(\partial_{r}f)^{2}],
\end{split}
\end{equation}
where the constant $C$ depends only on $M, r_{0}, r_{1}$ and $\Sigma_{0}$.
\end{proposition}

\subsection{The estimate for the spherical mean}
\label{sec:TheEstimateForTheSphericalMean}

Let $\psi_{0}$ denote the spherical mean of $\psi$, that is
\[\psi_{0}(v,r)=\int_{\mathbb{S}^{2}}\psi (v,r,\omega) \,d\omega,\]
where $\omega =(\theta,\phi)$ and $d\omega=\sin\theta d\theta\, d\phi$. We will prove the following proposition for $\psi_{0}$.

\begin{proposition} There is a constant $C>0$ that depends only on the mass parameter $M$ such that for spherically symmetric solutions $\psi_{0}$ to the wave equation on extremal Reissner--Nordstr\"{o}m which arise from initial data with bounded norm
\begin{equation}
D_{\Sigma_{0}}[\psi_{0}]: = \ \int_{\Sigma_{0}\cap\left\{r\leq R\right\}}\left[\frac{1}{(r-M)}\cdot \Big(\partial_{r}(r\psi_{0})\Big)^{2}\right]dr\,d\omega+\int_{\Sigma_{0}}\Big(J^{T}[\psi_{0}]\cdot \textbf{n}_{\Sigma_{0}}\Big)\, dg_{\Sigma_{0}}
\label{d0norm}
\end{equation}
 the following estimate holds
\begin{equation}
\begin{split}
\int_{\mathcal{A}}\Big[\psi_{0}^{2}+(\partial_{r}\psi_{0})^{2}+(\partial_{v}\psi_{0})^{2}\Big]\, dg_{\mathcal{A}}\leq  & \ C \cdot D_{\Sigma_{0}}[\psi].  
\end{split}
\label{prop1l0e1}
\end{equation}
Here the $J^{T}$ flux is as defined in Section \ref{background}.
\label{prop1l0}
\end{proposition}
\begin{proof}
We  apply the \textit{singular} vector field 
\begin{equation}
S=-\frac{1}{(r-M)}\cdot\partial_{r}
\label{defs}
\end{equation}
as our \textit{multiplier} vector field in the spacetime region $\mathcal{A}_{r_{0}}$ bounded by the hypersurfaces $\Sigma_{0}=\left\{v=0\right\}\cap\left\{r\leq R\right\}$,  $B_{R}=\left\{r=R\right\}$,  for some large $R>M$, and $H_{r_{0}}=\left\{u=u_{0}\right\}$, where $v_{0}=v(p)$ and the point $p$ is on the hypersurface $\Sigma_{0}$ such that $R>r(p)=r_{0}>M$. Here $u$ is the retarded null coordinate and $v$ is the advanced null coordinate (see Section \ref{background}). Clearly, we have that $H_{r_{0}}\rightarrow \mathcal{H}$ as $r_{0}\rightarrow M$. We therefore use that 
\begin{equation}
\int_{\mathcal{A}_{r_{0}}}\left(\frac{1}{(r-M)}\cdot \partial_{r}\psi_{0}\cdot \Box_{g}\psi_{0} \right)\cdot r^{2} dr\, dv\, d\omega=0.
\label{1l0}
\end{equation}

Since $\psi_{0}$ is spherically symmetric we have
\begin{equation}
\Box_{g}\psi_{0}=D\cdot\partial_{r}\partial_{r}\psi_{0}+2\partial_{v}\Big(H[\psi_{0}]\Big)+R\cdot\partial_{r}\psi_{0}=0
\label{wel0}
\end{equation}
where the partial derivatives are taken with respect to the ingoing Eddington--Finkelstein coordinates $(v,r,\theta,\phi)$ and
\begin{equation}
D=\left(\frac{r-M}{r}\right)^{2}, \ \ \ \ R=\frac{2}{r^{2}}(r-M), \ \ \ \  H[\psi_{0}]=\partial_{r}\psi_{0}+\frac{1}{r}\psi_{0}.
\label{defsl0}
\end{equation}
We therefore obtain that 
\begin{equation}
I_{1}+I_{2}+I_{3}=0,
\label{i}
\end{equation}
where 
\begin{align}
\label{i1}
I_{1}=&\int_{\mathcal{A}_{r_{0}}}(r-M)\cdot \partial_{rr}\psi_{0}\cdot \partial_{r}\psi_{0}\  dr\,dv\,d\omega,\\
\label{i2}
I_{2}=&\int_{\mathcal{A}_{r_{0}}}\frac{2r^{2}}{(r-M)}\cdot \partial_{r}\psi_{0}\cdot \partial_{v}\Big(H[\psi_{0}]\Big) \ dr\,dv\,d\omega,\\
\label{i3}
I_{3}=&\int_{\mathcal{A}_{r_{0}}}2(\partial_{r}\psi_{0})^{2}\  dr\,dv\,d\omega.
\end{align}
By integrating by parts\footnote{Note that $\partial_{r}$ is tangential to $\Sigma_{0}$.} with respect to $\partial_{r}$ in  $\mathcal{A}_{r_{0}}$ we obtain
\begin{equation}
\begin{split}
I_{1}=& \int_{\mathcal{A}_{r_{0}}}\left(\partial_{r}\left(\frac{1}{2}(r-M)\cdot (\partial_{r}\psi_{0})^{2}\right)^{2}-\frac{1}{2}(\partial_{r}\psi_{0})^{2}\right)\ dr\,dv\,d\omega\\
=& -\frac{1}{2}\int_{\mathcal{A}_{r_{0}}}(\partial_{r}\psi_{0})^{2}\ dr\,dv\,d\omega +\int_{B_{R}}\frac{1}{2}(r-M)\cdot (\partial_{r}\psi_{0})^{2}\ dv\, d\omega \\&-\int_{H_{r_{0}}}\frac{1}{2}(r-M)\cdot (\partial_{r}\psi_{0})^{2}\ dv\,d\omega. 
\end{split}
\label{i12}
\end{equation}
Note that \textbf{the coefficient of the spacetime integral on the right hand side has the wrong sign} and hence its precise value plays a fundamental role in our analysis. Specifically, it is crucial that \textit{the coefficient of the spacetime integral is stictly less than $2$ and hence this integral can be absorbed by the integral} $I_{3}$ (see equation \eqref{i3}). 

By integrating by parts\footnote{Note that $\partial_{v}$ is tangential to $B_{R}$.} with respect to $\partial_{v}$ in  $\mathcal{A}_{r_{0}}$, and using equations \eqref{defsl0}, \eqref{i3} and that 
\begin{equation}
dr=\frac{(r-M)^{2}}{2r^{2}}dv \ \ : \ \  \text{ along } H_{r_{0}}  
\label{drdv}
\end{equation}
 we obtain
\begin{equation}
\begin{split}
I_{2}&= I_{4}+ \int_{\mathcal{A}_{r_{0}}}\partial_{v}\left(\frac{r^{2}}{(r-M)}\cdot (\partial_{r}\psi_{0})^{2}\right)\ dv\, dr\, d\omega\\
&= I_{4}+\int_{H_{r_{0}}}\frac{r^{2}}{(r-M)}\cdot (\partial_{r}\psi_{0})^{2}\, dr\,d\omega-\int_{\Sigma_{0}\cap\left\{r_{0}\leq r\leq R\right\}}\frac{r^{2}}{(r-M)}\cdot (\partial_{r}\psi_{0})^{2}\, dr\,d\omega\\
&\overset{\eqref{drdv}}{=} I_{4}+\int_{H_{r_{0}}}\frac{1}{2}\cdot (r-M) (\partial_{r}\psi_{0})^{2}\, dv\,d\omega-\int_{\Sigma_{0}\cap\left\{r_{0}\leq r\leq R\right\}}\frac{r^{2}}{(r-M)}\cdot (\partial_{r}\psi_{0})^{2}\, dr\,d\omega,
\end{split}
\label{i22}
\end{equation}
where 
\begin{equation}
I_{4}= \int_{\mathcal{A}_{r_{0}}}\frac{2r}{(r-M)}\cdot (\partial_{v}\psi_{0})\cdot(\partial_{r}\psi_{0}) \, dr\, dv\, d\omega. 
\label{i4}
\end{equation}
Clearly, all the terms are regular apart from the term $I_{4}$ which is singular when we take the limit $r_{0}\rightarrow M$. We will show that \textbf{in view of the special structure of the geometry of degenerate horizons we are able to bound this integral  in terms of a weighted norm of the initial data on $\Sigma_{0}$ only}  (see also the discussion in Section \ref{sec:Discussion}). Indeed, by integrating by parts with respect to $\partial_{v}$ we obtain
\begin{equation}
\begin{split}
I_{4}=&I_{5}+\int_{\mathcal{A}_{r_{0}}}\partial_{v}\left(\frac{2r}{(r-M)}\cdot \psi_{0}\cdot \partial_{r}\psi_{0}\right)\, dv\,dr\,d\omega\\
=& I_{5}+\int_{H_{r_{0}}}\frac{2r}{(r-M)}\cdot \psi_{0}\cdot (\partial_{r}\psi_{0})\, dr\, d\omega -\int_{\Sigma_{0}\cap\left\{r_{0}\leq r\leq R\right\}}\frac{2r}{(r-M)}\cdot\psi_{0}\cdot  \partial_{r}\psi_{0}\, dr\,d\omega \\
\overset{\eqref{drdv}}{=}&
I_{5}+\int_{H_{r_{0}}}\frac{(r-M)}{r}\cdot \psi_{0}\cdot (\partial_{r}\psi_{0})\, dv\, d\omega -\int_{\Sigma_{0}\cap\left\{r_{0}\leq r\leq R\right\}}\frac{2r}{(r-M)}\cdot\psi_{0}\cdot  \partial_{r}\psi_{0}\, dr\,d\omega 
\end{split}
\label{i41}
\end{equation}
where
\begin{equation}
I_{5}=\int_{\mathcal{A}_{r_{0}}}\frac{r}{(r-M)}\cdot\psi_{0}\cdot (-2\partial_{v}\partial_{r}\psi_{0})\, dv\,dr\,d\omega.
\label{i5}
\end{equation}
The wave equation \eqref{wel0} and the expression \eqref{defsl0}  for $H[\psi_{0}]$ yield
\begin{equation}
\begin{split}
I_{5}=&\int_{\mathcal{A}_{r_{0}}}\frac{r}{(r-M)}\cdot \psi_{0}\cdot \!\!\left[\frac{(r-M)^{2}}{r^{2}}\cdot \partial_{r}\partial_{r}\psi_{0}+\frac{2}{r}\cdot \partial_{v}\psi_{0}+\frac{2}{r^{2}}\cdot (r-M)\cdot \partial_{r}\psi_{0}\right]dr\,dv\,d\omega\\
=& I_{6}+I_{7}+I_{8},
\end{split}
\label{i51}
\end{equation}
where 
\begin{align}
\label{i6}
I_{6}=&\int_{\mathcal{A}_{r_{0}}}\frac{2}{(r-M)}\cdot\psi_{0}\cdot \partial_{v}\psi_{0} \, dr\,dv\,d\omega,\\
\label{i7}
I_{7}=&\int_{\mathcal{A}_{r_{0}}}\frac{2}{r}\cdot \psi_{0}\cdot\partial_{r}\psi_{0} \, dr\,dv\,d\omega,\\
\label{i8}
I_{8}=&\int_{\mathcal{A}_{r_{0}}}\frac{(r-M)}{r}\cdot\psi_{0}\cdot\partial_{r}\partial_{r}\psi_{0} \, dr\,dv\,d\omega.
\end{align}
Furthermore,
\begin{equation}
\begin{split}
I_{6}=&\int_{\mathcal{A}_{r_{0}}}\partial_{v}\left(\frac{1}{(r-M)}\cdot\psi_{0}^{2}\right)dr\,dv\,d\omega\\
=&\int_{H_{r_{0}}}\frac{1}{(r-M)}\cdot\psi_{0}^{2}\, dr\,d\omega -\int_{\Sigma_{0}\cap\left\{r_{0}\leq r\leq R\right\}}\frac{1}{(r-M)}\cdot\psi_{0}^{2}\, dr\,d\omega\\
\overset{\eqref{drdv}}{=}&\int_{H_{r_{0}}}\frac{(r-M)}{2r^{2}}\cdot\psi_{0}^{2}\, dv\,d\omega -\int_{\Sigma_{0}\cap\left\{r_{0}\leq r\leq R\right\}}\frac{1}{(r-M)}\cdot\psi_{0}^{2}\, dr\,d\omega.
\end{split}
\label{i61}
\end{equation}
Similarly, we obtain
\begin{equation}
\begin{split}
I_{7}=&\int_{\mathcal{A}_{r_{0}}}\frac{1}{r}\cdot\partial_{r}\psi_{0}^{2}\, dr\, dv\, d\omega=\int_{\mathcal{A}_{r_{0}}}\left[\partial_{r}\left(\frac{\psi_{0}^{2}}{r}\right)+\frac{1}{r^{2}}\cdot\psi_{0}^{2}\right] dr\, dv\, d\omega\\
=& \int_{B_{R}}\frac{\psi_{0}^{2}}{r}\, dv\, d\omega -\int_{H_{r_{0}}}\frac{\psi_{0}^{2}}{r}\, dv\,d\omega +\int_{\mathcal{A}_{r_{0}}}\frac{1}{r^{2}}\cdot \psi_{0}^{2}\, dr\, dv\, d\omega\\
=& \int_{B_{R}}\frac{\psi_{0}^{2}}{r}\, dv\, d\omega -I_{9}+\int_{\mathcal{A}_{r_{0}}}\frac{1}{r^{2}}\cdot \psi_{0}^{2}\, dr\, dv\, d\omega,
\end{split}
\label{i71}
\end{equation}
where 
\begin{equation}
I_{9}=\int_{H_{r_{0}}}\frac{\psi_{0}^{2}}{r}\, dv\,d\omega. 
\label{i9}
\end{equation}
Observe that the middle integral $I_{9}$ above has the wrong sign in the expression for $I_{7}$ in \eqref{i71}. This will be later remedied using an appropriate Hardy inequality.

Regarding the integral $I_{8}$ we obtain the following
\begin{equation}
\begin{split}
I_{8}=&\int_{\mathcal{A}_{r_{0}}}\left[\partial_{r}\left(\frac{(r-M)}{r}\cdot\psi_{0}\cdot\partial_{r}\psi_{0}\right)-\frac{(r-M)}{r}\cdot (\partial_{r}\psi_{0})^{2}-\frac{M}{r^{2}}\cdot \psi_{0}\cdot \partial_{r}\psi_{0}\right]dr\,dv\, d\omega\\
=& \int_{B_{R}}\frac{(r-M)}{r}\cdot\psi_{0}\cdot\partial_{r}\psi_{0}\, dv\,d\omega -\int_{H_{r_{0}}}\frac{(r-M)}{r}\cdot\psi_{0}\cdot\partial_{r}\psi_{0}\, dv\,d\omega\\&-\int_{\mathcal{A}_{r_{0}}}\frac{(r-M)}{r}\cdot (\partial_{r}\psi_{0})^{2}\, dr\,dv\,d\omega -I_{10},
\end{split}
\label{i81}
\end{equation}
where 
\begin{equation}
\begin{split}
I_{10}=&\int_{\mathcal{A}_{r_{0}}}\frac{M}{2r^{2}}\cdot\partial_{r}\psi_{0}^{2}\, dr\,dv\,d\omega= \int_{\mathcal{A}_{r_{0}}}
\left[\partial_{r}\left(\frac{M}{2r^{2}}\cdot\psi_{0}^{2}\right)+\frac{M}{r^{3}}\cdot\psi_{0}^{2}\right] dr\,dv\,d\omega\\
=&\int_{B_{R}}\frac{M}{2r^{2}}\cdot \psi_{0}^{2}\,dv\,d\omega -\int_{H_{r_{0}}}\frac{M}{2r^{2}}\cdot\psi_{0}^{2}\,dv\,d\omega+\int_{\mathcal{A}_{r_{0}}}\frac{M}{r^{3}}\cdot\psi_{0}^{2}\, dr\,dv\,d\omega.
\label{i10}
\end{split}
\end{equation}
Therefore, by   using equations \eqref{i}, \eqref{i3}, \eqref{i12}, \eqref{i22}, \eqref{i41},  \eqref{i51},  \eqref{i61},  \eqref{i71},  \eqref{i81}, \eqref{i10} and grouping all the integral terms in $I_{1},I_{2},I_{3}$ we obtain
\begin{equation}
\begin{split}
0=& \int_{\mathcal{A}_{r_{0}}}\left[-\frac{1}{2}\cdot(\partial_{r}\psi_{0})^{2}+\frac{1}{r^{2}}\cdot\psi_{0}^{2}-\frac{(r-M)}{r}\cdot(\partial_{r}\psi_{0})^{2}-\frac{M}{r^{3}}\cdot \psi^{2}+2(\partial_{r}\psi_{0})^{2}\right]dr\,dv\,d\omega \\
&+\int_{H_{r_{0}}}\left[\underbrace{-\frac{1}{2}(r-M)\cdot(\partial_{r}\psi_{0})^{2}}_{1}+\underbrace{\frac{(r-M)}{r}\cdot\psi_{0}\cdot\partial_{r}\psi_{0}}_{2}+\frac{(r-M)}{2r^{2}}\cdot\psi_{0}^{2}\right]dv\,d\omega \\ 
&+\int_{H_{r_{0}}}\left[-\frac{1}{r}\cdot\psi_{0}^{2}+\underbrace{\frac{1}{2}(r-M)\cdot(\partial_{r}\psi_{0})^{2}}_{1}+\frac{M}{2r^{2}}\cdot\psi_{0}^{2}\underbrace{-\frac{(r-M)}{r}\cdot\psi_{0}\cdot\partial_{r}\psi_{0}}_{2}\right]dv\,d\omega \\
&+ \int_{\Sigma_{0}\cap\left\{r_{0}\leq r\leq R\right\}}\left[-\frac{r^{2}}{(r-M)}\cdot(\partial_{r}\psi_{0})^{2}-\frac{2r}{(r-M)}\cdot\psi_{0}\cdot\partial_{r}\psi_{0}-\frac{1}{(r-M)}\cdot\psi_{0}^{2}\right]dr\,d\omega \\
&+\int_{B_{R}}\Big[K_{R}[\psi]\Big]dv\,d\omega,
\end{split}
\label{l0f1}
\end{equation}
where 
\begin{equation}
K_{R}[\psi_{0}]=\frac{1}{2}\cdot (r-M)\cdot(\partial_{r}\psi_{0})^{2}+\left(\frac{2r-M}{2r^{2}}\right)\cdot \psi_{0}^{2}+\frac{(r-M)}{r}
\cdot\psi_{0}\cdot\partial_{r}\psi_{0}.
\label{krpsi}
\end{equation}
Hence, by noting all terms that cancel out,  we have established that
\begin{equation}
\begin{split}
& \int_{\mathcal{A}_{r_{0}}}\left[\left(\frac{3}{2}-\frac{(r-M)}{r}\right)\cdot (\partial_{r}\psi_{0})^{2}+\left(\frac{r-M}{r^{3}}\right)\cdot \psi_{0}^{2}\right]dr\,dv\,d\omega \\
=& \int_{\Sigma_{0}\cap\left\{r_{0}\leq r\leq R\right\}}\left[\frac{r^{2}}{(r-M)}\cdot\left(\partial_{r}\psi_{0}+\frac{1}{r}\psi_{0}\right)^{2}\right]dr\,d\omega \\
&+\int_{H_{r_{0}}}\left[-\frac{1}{2r}\cdot\psi_{0}^{2}\right]dv\,d\omega -\int_{B_{R}}\Big[K_{R}[\psi_{0}]\Big]dv\,d\omega. 
\end{split}
\label{l0f2}
\end{equation}
By the second Hardy inequality of Section \ref{sec:HardyInequalities} we have that there is an $\epsilon_{1}(M)>0$ such that for all $0<\epsilon<\epsilon_{1}$ we have
\begin{equation}
\int_{H_{r_{0}}}\left[\frac{1}{2r}\cdot\psi_{0}^{2}\right]dv\,d\omega\leq \epsilon  \int_{\mathcal{A}_{r_{0}}}(\partial_{r}\psi_{0})^{2}\, dr\, dv\,d\omega +
\frac{1}{\epsilon} \int_{\mathcal{A}_{r_{0}}}\psi_{0}^{2}\, dr\, dv\, d\omega +\frac{1}{\epsilon}\int_{B_{R}}E_{1}[\psi_{0}]\, dv\, d\omega,
\label{hb1}
\end{equation}
where 
\begin{equation}
E_{1}[\psi_{0}]\sim \psi_{0}^{2}
\label{e1}
\end{equation}
where the constants in $\sim$ depend only on $M$ (recall that $M<R<2M$). 

By the first Hardy inequality of Section \ref{sec:HardyInequalities} we have that there is an $\epsilon_{2}(M)>0$ such that for all $0<\epsilon<\epsilon_{2}$ we have
\begin{equation}
\frac{1}{\epsilon} \int_{\mathcal{A}_{r_{0}}}\psi_{0}^{2}\, dr\, dv\, d\omega\leq 
\frac{1}{\epsilon^{2}} \int_{\mathcal{A}_{r_{0}}}(r-M)^{2}\cdot(\partial_{r}\psi_{0})^{2}\, dr\, dv\, d\omega+\frac{1}{\epsilon^{2}}\int_{B_{R}}E_{2}[\psi_{0}]\, dv\, d\omega,
\label{hb2}
\end{equation}
where 
\begin{equation}
E_{2}[\psi_{0}]\sim \psi_{0}^{2}
\label{e1}
\end{equation}
where the constants in $\sim$ depend only on $M$.
Noting that for all $r\geq M$ we have
\[ \frac{3}{2}-\frac{(r-M)}{r}>\frac{1}{2}\]
and using \eqref{l0f2} and  $\epsilon$ in \eqref{hb1}, \eqref{hb2} sufficiently small,  we obtain
\begin{equation}
\begin{split}
\int_{\mathcal{A}_{r_{0}}}\frac{1}{2}\cdot (\partial_{r}\psi_{0})^{2}\, dr\, dv\, d\omega \leq &\int_{\Sigma_{0}\cap\left\{r_{0}\leq r\leq R\right\}}\left[\frac{1}{(r-M)}\cdot \Big(\partial_{r}(r\psi_{0})\Big)^{2}\right]dr\,d\omega\\
&+  \int_{\mathcal{A}_{0}}\left[\epsilon+\frac{(r-M)^{2}}{\epsilon^{2}}\right](\partial_{r}\psi_{0})^{2}\, dr\, dv\,d\omega\\
&+\int_{B_{R}}E_{3}[\psi_{0}]\,dv\, d\omega,
\end{split}
\label{cf1}
\end{equation}
where
\begin{equation}
E_{3}[\psi_{0}]=E_{1}[\psi_{0}]+E_{2}[\psi_{0}]-K_{R}[\psi_{0}]\sim \psi_{0}^{2}+(\partial_{r}\psi_{0})^{2}.
\label{eq:}
\end{equation}
We now choose $\epsilon$ such that
\begin{equation}
\epsilon=\min\left\{\frac{1}{16}, \epsilon_{1},\epsilon_{2}\right\},
\label{e}
\end{equation}
where $\epsilon_{1},\epsilon_{2}$ are the constants of the Hardy inequalities \eqref{hb1}, \eqref{hb2}, respectively. 
Clearly, with this choice $\epsilon$ depends only on $M$. Recalling that in region $\mathcal{A}_{r_{0}}$ we have $r\leq R$, we impose on $R$ the condition 
\begin{equation*}
\frac{1}{\epsilon^{2}}\cdot \frac{(R-M)^{2}}{M^{2}}\leq\frac{1}{16}
\end{equation*}
which implies 
\begin{equation}
\frac{R}{M}\leq 1+\frac{\epsilon}{4},
\label{R}
\end{equation}
where $\epsilon$ is given by \eqref{e}. With these conditions for $\epsilon,R$, estimate \eqref{cf1} yields the following
\begin{equation}
\begin{split}
\int_{\mathcal{A}_{r_{0}}}\frac{1}{4}\cdot (\partial_{r}\psi_{0})^{2}\, dr\, dv\, d\omega \leq &\int_{\Sigma_{0}\cap\left\{r_{0}\leq r\leq R\right\}}\left[\frac{1}{(r-M)}\cdot \Big(\partial_{r}(r\psi_{0})\Big)^{2}\right]dr\,d\omega\\
&+\int_{B_{R}}E_{3}[\psi_{0}]\,dv\, d\omega,
\end{split}
\label{cf2}
\end{equation}
We finally need to bound the boundary integral over $B_{R}$. In view of the degenerate Morawetz estimate of \cite{aretakis1} we have that 
\begin{equation*}
\int_{\left\{M+\epsilon\frac{M}{8}\leq r\leq M+\epsilon\frac{M}{4}\right\}}\Big[(\partial_{r}\psi_{0})^{2}+\psi_{0}^{2}\Big] dr\,dv\,d\omega \leq C_{\tilde{\epsilon}}\int_{\Sigma_{0}}\Big(J^{T}[\psi_{0}]\cdot \textbf{n}_{\Sigma_{0}}\Big)\, dg_{\Sigma_{0}}.
\end{equation*}
Hence, by the averaging principle, there is a value 
\begin{equation}
\tilde{R}\in \left[M+\epsilon\frac{M}{8},M+\epsilon\frac{M}{4}\right]
\label{R1}
\end{equation}
such that 
\begin{equation}
\int_{\left\{r=\tilde{R}\right\}}\Big[(\partial_{r}\psi_{0})^{2}+\psi_{0}^{2}\Big] dv\,d\omega \leq C\int_{\Sigma_{0}}\Big(J^{T}[\psi_{0}]\cdot \textbf{n}_{\Sigma_{0}}\Big)\, dg_{\Sigma_{0}}.
\label{averaging}
\end{equation}
Note that $C$ depends only on $M$ since $\epsilon$ has already been chosen in \eqref{e}. Therefore, if we define \[R:=\tilde{R}\]
then 
\eqref{cf2} becomes
\begin{equation}
\begin{split}
\int_{\mathcal{A}_{r_{0}}}\frac{1}{4}\cdot (\partial_{r}\psi_{0})^{2}\, dr\, dv\, d\omega \leq &\int_{\Sigma_{0}\cap\left\{r_{0}\leq r\leq R\right\}}\left[\frac{1}{(r-M)}\cdot \Big(\partial_{r}(r\psi_{0})\Big)^{2}\right]dr\,d\omega\\
&+C\int_{\Sigma_{0}}\Big(J^{T}[\psi_{0}]\cdot \textbf{n}_{\Sigma_{0}}\Big)\, dg_{\Sigma_{0}}.
\end{split}
\label{cf3}
\end{equation}
Clearly, all the constants are independent of the constant $r_{0}$ in the definition of the spacetime region $\mathcal{A}_{r_{0}}$. Therefore, by taking $r_{0}\rightarrow M$ in \eqref{cf3} we obtain Proposition \eqref{prop1l0}. The bound on the zeroth order term follow from the first Hardy inequality \eqref{1hardy}.

\end{proof}

\subsection{The estimate for angular frequencies $\ell\geq 1$}
\label{sec:TheEstimateOnAngularFrequenciesEllGeq1}

For the projection on angular frequencies $\ell\geq 1$ we apply \textit{regular multiplier} \textbf{and} \textit{commutator} vector fields in the spacetime region $\mathcal{R}(0,\tau)$ bounded by the hypersurfaces $\Sigma_{0}$ and $\Sigma_{\tau}$.

\begin{proposition}\label{2penes2}
Let $\psi$ be a solution of the linear wave equation $\Box_g \psi = 0$. Then for any $\tau > 0$  we have that for the part of $\psi$ that is localized in angular frequencies $ \meg 1$ the following estimate holds true:
\begin{equation}\label{2penese2}
\int_{\si_{\tau} \cap \cala_{0}^{\tau}} \left(  J^P [ \partial_{r}\psi_{\meg 1}] \cdot \textbf{n}_{\si_{0}} \right) dg_{\si_{\tau}} + 
\end{equation}
$$ + \int_{\mathcal{A}_{0}^{\tau}} \Big( (\partial_{v}\partial_{r}\psi_{\meg 1} )^2 + (r-M)^{2} \cdot( \partial_{r}\partial_{r}\psi_{\meg 1} )^2 + |  \slashed{\nabla} \partial_r \psi_{\meg 1} |^2 \Big) dg_{\mathcal{A}_{0}^{\tau}} \lesssim  $$
$$ \lesssim \sum_{l=0,1} \int_{\si_{0}}  \left( J^P [ \partial_{v}^l \psi_{\meg 1} ] \cdot \textbf{n}_{\si_{0}} \right) dg_{\si_{0}} + \int_{\si_{0} \cap \cala_{0}^{\tau}}  \left( J^P [ \partial_{r}\psi_{\meg 1} ] \cdot \textbf{n}_{\si_{0}} \right) dg_{\si_{0}} , $$
where 
$$ \mathcal{A}_{0}^{\tau} = \rrr (0 , \tau )\cap \mathcal{A},$$ where $\mathcal{A}$ is as defined in Section \ref{background}. 
\end{proposition}
\begin{proof}
We consider the equation for $\partial_{r}\psi_{\meg 1}$, we have that:
\begin{equation}\label{yp}
\Box_g (\partial_{r}\psi_{\meg 1} ) = D' \partial_{r}\partial_{r}\psi_{\meg 1} + \frac{2}{r^2} \partial_{v}\psi_{\meg 1} - R' \partial_{r}\psi_{\meg 1} + \frac{2}{r} \slashed{\Delta} \psi_{\meg 1} .
\end{equation}
Now consider the vector field:
$$ L_P = f^v (r) \partial_v + f^r (r) \partial_r , $$
where $f^v$ and $f^r$ are smooth functions satisfying
$$ f^v \simeq 1, \quad \partial_r f^v \simeq \frac{1}{\sigma}, \quad f^r = - M \sqrt{D}, \quad \partial_r f^r = -\frac{M^2}{r^2}, $$
close to the horizon (in the region $\mathcal{A}$ where $r_0$ is chosen to be very close to $M$), with $f^v \equiv 1$, $f^r \equiv 0$ in $r \meg r_1$ for some $r_0 < r_1 < 2M $, and where $\sigma > 0$ is chosen to be small. Note that $L_{P}\sim P$, where $P$ is the vector field defined in Section \ref{flux}.

Applying Stokes' Theorem for $J^{L_P} [\partial_{r}\psi_{\meg 1}]$ we have that the following spacetime terms:
$$ K^{L_P} [\partial_{r}\psi_{\meg 1}] + \mathcal{E}^{L_P} [\partial_{r}\psi_{\meg 1}] = H_1 (\partial_{v}\partial_{r}\psi_{\meg 1})^2 + H_2 (\partial_{r}\partial_{r}\psi_{\meg 1} )^2 + H_3 |\slashed{\nabla} \partial_{r}\psi_{\meg 1} |^2 + $$ $$ + H_4 (\partial_{v}\partial_{r}\psi_{\meg 1} ) \cdot (\partial_{v}\psi_{\meg 1} ) + H_5 (\partial_{v}\partial_{r}\psi_{\meg 1} ) \cdot (\partial_{r}\psi_{\meg 1} ) + H_6 (\partial_{r}\partial_{r}\psi_{\meg 1} ) \cdot (\partial_{v}\psi_{\meg 1} ) + H_7 (\partial_{v}\partial_{r}\psi_{\meg 1} ) \cdot (\slashed{\Delta} \psi_{\meg 1} ) + $$ $$ + H_8 (\partial_{r}\partial_{r}\psi_{\meg 1} ) \cdot (\slashed{\Delta} \psi_{\meg 1} ) + H_9 (\partial_{v}\partial_{r}\psi_{\meg 1} ) \cdot (\partial_{r}\partial_{r}\psi_{\meg 1} ) + H_{10} (\partial_{r}\partial_{r}\psi_{\meg 1} ) \cdot (\partial_{r}\psi_{\meg 1} ) , $$
where close to the horizon
$$ H_1 = (\partial_r f^v ) \simeq \frac{1}{\sigma}, \quad H_2 = \frac{D (\partial_r f^r )}{2} - \frac{D f^r}{r} - \frac{3D' f^r}{2} =  \frac{5M^2 D}{2r^2} + \frac{M D^{3/2}}{r} , $$ $$ H_3 = -\frac{1}{2} (\partial_r f^r ) = \frac{M^2}{2r^2}, \quad H_4 = \frac{2f^v}{r^2} \simeq \frac{2}{r^2}, \quad H_5 = - f^v R' \simeq -R' , $$ $$ H_6 = \frac{2f^r}{r^2} = - \frac{2M\sqrt{D}}{r^2} , \quad H_7 = \frac{2f^v}{r} \simeq \frac{2}{r} , \quad H_8 = \frac{2f^r}{r} = - \frac{2M\sqrt{D}}{r} , $$ $$ H_9 = D (\partial_r f^v ) - D' f^v - \frac{2f^r}{r} \simeq \frac{D}{\sigma} - D' + \frac{2M\sqrt{D}}{r} , \quad H_{10} = M \sqrt{D} R' . $$ 
Here the functions $D(r),R(r)$ are given by \eqref{defsl0}. 
We will not deal with the terms away from the horizon since they can be bounded by a \textit{degenerate} Morawetz estimate away from the photon sphere for $J^T [\partial_{v}\psi_{\meg 1}]$.

The boundary fluxes are given by equation \eqref{fluxestpn} (recall that $L_{P}\sim P$). 

We next bound the spacetime terms. Note that $H_{1},H_{2},H_{3}$ are positive, have desired asymptotics at the horizon and are thus the most relevant terms. We next consider the terms $H_4 - H_{10}$, which will be shown to be error terms relative to $H_{1},H_{2},H_{3}$. 

\textbf{$H_4$}: We have that
$$ \int_{\cala_{0}^{\tau}} H_4 (\partial_{v}\partial_{r}\psi_{\meg 1} ) \cdot (\partial_{v}\psi_{\meg 1} ) dg_{\cala} \simeq \int_{\cala_{0}^{\tau}} \frac{2}{r^2} (\partial_{v}\partial_{r}\psi_{\meg 1} ) \cdot (\partial_{v}\psi_{\meg 1} ) dg_{\cala} \mik $$ $$ \mik \beta \int_{\cala_{0}^{\tau}}  (\partial_{v}\partial_{r}\psi_{\meg 1} )^2  dg_{\cala} + \frac{1}{\beta} \int_{\cala_{0}^{\tau}}  (\partial_{v}\psi_{\meg 1} )^2 dg_{\cala} , $$
and now we absorb the first term in the right hand by a choice of a $\beta = \beta (M)$ that is small enough, and we bound the second term by the Morawetz estimate.

Note that the $H_4$ term introduced only a $\beta$ loss (for $\beta$ very small) from the $H_1$ term.

\textbf{$H_5$}: We have that
\begin{equation}\label{h5}
 \int_{\cala_{0}^{\tau}} H_5 (\partial_{v}\partial_{r}\psi_{\meg 1} ) \cdot (\partial_{r}\psi_{\meg 1} ) dg_{\cala} \simeq \int_{\cala_{0}^{\tau}} -R' (\partial_{v}\partial_{r}\psi_{\meg 1} ) \cdot (\partial_{r}\psi_{\meg 1} ) dg_{\cala} \mik 
\end{equation}
$$ \mik \frac{1}{\beta} \int_{\cala_{0}^{\tau}}  (\partial_{v}\partial_{r}\psi_{\meg 1} )^2 dg_{\cala} + \beta \int_{\cala_{0}^{\tau}} (\partial_{r}\psi_{\meg 1} )^2 dg_{\cala} , $$
where $\beta = \beta (M)$ since in $\cala$ we have that 
$$R' = D'' + \frac{2D'}{r} - \frac{2D}{r^2} = \dfrac{MD'}{\sqrt{D} r^2} - \dfrac{2M\sqrt{D}}{r^3} + \frac{2D'}{r} - \frac{2D}{r^2} \simeq \frac{2M^2}{r^4} , $$
 and it is chosen to be small enough but much bigger than $\sigma$, so that the first term of \eqref{h5} can be absorbed in the right hand side, while for the second one we apply the third Hardy inequality \eqref{3hardy}
 $$ \beta \int_{\cala_{0}^{\tau}} (\partial_{r}\psi_{\meg 1} )^2 dg_{\cala} \mik $$ $$ \mik C \beta \int_{\rrr(0 , \tau ) \cap \{ r_0 \mik r_1 < 2M\}} (\partial_{r}\psi_{\meg 1} )^2 dg_{\rrr} + C \beta \int_{\cala_{0}^{\tau}} D \left( (\partial_{v}\partial_{r}\psi_{\meg 1} )^2 + (\partial_{r}\partial_{r}\psi_{\meg 1} )^2 \right) dg_{\cala} + $$ $$ + C \beta \int_{\rrr(0 , \tau ) \cap \{ r_0 \mik r_1 < 2M\}} D \left( (\partial_{v}\partial_{r}\psi_{\meg 1} )^2 + (\partial_{r}\partial_{r}\psi_{\meg 1} )^2 \right) dg_{\rrr} , $$
 where the first and the third term of the above estimate can be bounded by the Morawetz estimate for $\psi_{\meg 1}$ and $\partial_{v}\psi_{\meg 1}$ respectively, while the second one can be absorbed in the right hand side as 
 $$ C\beta D \leq \frac{H_1}{10} \mbox{ and } C\beta D \leq \frac{H_2}{10}  $$
in $\mathcal{A}_{0}^{\tau}$, 
 due to the higher degeneracy of $D$ on the horizon compared to the other terms for the first inequality, and due to the smallness of $\beta$ for the second one.
 
 Note that the $H_5$ term introduced only a $\beta$ loss (for $\beta$ very small) from the $H_1$ and $H_2$ terms.
 
\textbf{$H_6$}: We have that
$$ \int_{\cala_{0}^{\tau}} H_6 (\partial_{r}\partial_{r}\psi_{\meg 1} ) \cdot (\partial_{v}\psi_{\meg 1} ) dg_{\cala} = -\int_{\cala_{0}^{\tau}} \frac{2M\sqrt{D}}{r^2} (\partial_{r}\partial_{r}\psi_{\meg 1} ) \cdot (\partial_{v}\psi_{\meg 1} ) dg_{\cala} \mik $$ $$ \mik \beta \int_{\cala_{0}^{\tau}} D(\partial_{r}\partial_{r}\psi_{\meg 1} )^2 dg_{\cala} + \frac{1}{\beta} \int_{\cala_{0}^{\tau}} (\partial_{v}\psi_{\meg 1} )^2 dg_{\cala} , $$
for $\beta = \beta (M)$ small enough, and the first term can be absorbed in the right hand side, while the second one is bounded by the Morawetz estimate.

Note that $H_6$ term introduced only a $\beta$ loss (for $\beta$ very small) from the $H_1$ term.

\textbf{$H_7$}: We have that
\begin{equation}\label{h7}
\int_{\cala_{0}^{\tau}} H_7 (\partial_{v}\partial_{r}\psi_{\meg 1} ) \cdot (\slashed{\Delta} \psi_{\meg 1} ) dg_{\cala} \simeq  \int_{\cala_{0}^{\tau}} \frac{2}{r} (\partial_{v}\partial_{r}\psi_{\meg 1} ) \cdot (\slashed{\Delta} \psi_{\meg 1} ) dg_{\cala} = 
\end{equation}
 $$ =\int_{\si_{0}} \frac{2}{r} (\partial_{r}\psi_{\meg 1} ) \cdot (\slashed{\Delta} \psi_{\meg 1}  ) \cdot (\partial_v \cdot\textbf{n}_{\si} )  dg_{\si} - \int_{\si_{\tau}} \frac{2}{r} (\partial_{r}\psi_{\meg 1} ) \cdot (\slashed{\Delta} \psi_{\meg 1}  ) \cdot (\partial_v\cdot \textbf{n}_{\si} ) dg_{\si} - $$ $$ - \int_{\cala_{0}^{\tau}} \frac{2}{r}  (\partial_{r}\psi_{\meg 1} ) \cdot (\partial_v \slashed{\Delta} \psi_{\meg 1} ) dg_{\cala} ,$$
by an application of Stokes' Theorem.
 
For the last term of \eqref{h7} we have after integrating by parts on the sphere that
$$ - \int_{\cala_{0}^{\tau}} \frac{2}{r}  (\partial_{r}\psi_{\meg 1} ) \cdot (\partial_{v}\slashed{\Delta} \psi_{\meg 1} ) dg_{\cala} = \int_{\cala_{0}^{\tau}} \frac{2}{r}  \langle (\slashed{\nabla} \partial_{r}\psi_{\meg 1} ) ,  (\slashed{\nabla} \partial_{v}\psi_{\meg 1} ) \rangle dg_{\cala} \mik $$ $$ \mik \beta_1 \int_{\cala_{0}^{\tau}} |\slashed{\nabla} \partial_{r}\psi_{\meg 1} |^2 dg_{\cala} + \frac{1}{\beta_1} \int_{\cala_{0}^{\tau}} |\slashed{\nabla} \partial_{v}\psi_{\meg 1} |^2 dg_{\cala} ,$$
for some $\beta_1 = \beta_1 (M)$ that is chosen to be small enough so that the first term can be absorbed from the right hand side (from $H_{3}$), while the second term can be bounded by the $T$-flux of $\partial_{v}\psi_{\meg 1}$.

For the second term of \eqref{h7} we have after integrating by parts on the sphere that
$$ - \int_{\si_{\tau}} \frac{2}{r} (\partial_{r}\psi_{\meg 1} ) \cdot (\slashed{\Delta} \psi_{\meg 1}  ) \cdot (\partial_v \cdot\textbf{n}_{\si} ) dg_{\si} = \int_{\si_{\tau}} \frac{2}{r} \langle \slashed{\nabla} \partial_{r}\psi_{\meg 1} , \slashed{\nabla} \psi_{\meg 1}  \rangle \cdot (\partial_v \cdot \textbf{n}_{\si} ) dg_{\si} \mik $$ $$ \mik \beta_2 \int_{\si_{\tau}} |\slashed{\nabla} \partial_{r}\psi_{\meg 1} |^2 dg_{\si} + \frac{1}{\beta_2} \int_{\si_{\tau}} |\slashed{\nabla} \psi_{\meg 1} |^2 dg_{\si} , $$
where $\beta_2 = \beta_2 (M)$ is chosen to be small enough so that the first term can be absorbed from the flux term of the energy identity for for $J^{L_{P}}[\partial_{r}\psi_{\geq 1}]$ (recall that $L_{p}\sim P$ and hence the fluxes are given by \eqref{fluxestpn}), while the second term can be bounded by the $T$-flux of $\psi_{\meg 1}$. The first term of \eqref{h7} can be treated in a similar manner.

Note that the $H_7$ term introduced only a $\beta$ loss (for $\beta$ very small) from the $H_3$ term.

\textbf{$H_8$}: We have that
$$ \int_{\cala_{0}^{\tau}} H_8 (\partial_{r}\partial_{r}\psi_{\meg 1} ) \cdot (\slashed{\Delta} \psi_{\meg 1} ) dg_{\cala} = -\int_{\cala_{0}^{\tau}} \dfrac{2M\sqrt{D}}{r} (\partial_{r}\partial_{r}\psi_{\meg 1} ) \cdot (\slashed{\Delta} \psi_{\meg 1} ) dg_{\cala} $$ and so 
\begin{equation}\label{h8}
  -\int_{\cala_{0}^{\tau}} \dfrac{2M\sqrt{D}}{r} (\partial_{r}\partial_{r}\psi_{\meg 1} ) \cdot (\slashed{\Delta} \psi_{\meg 1} ) dg_{\cala} = 
\end{equation} 
  $$ = - \int_{\si_{0}} \dfrac{2M\sqrt{D}}{r} (\partial_{r}\psi_{\meg 1} ) \cdot (\slashed{\Delta} \psi_{\meg 1} ) \cdot (\partial_r \cdot \textbf{n}_{\si} )  dg_{\si} + \int_{\si_{\tau}} \dfrac{2M\sqrt{D}}{r} (\partial_{r}\psi_{\meg 1} ) \cdot (\slashed{\Delta} \psi_{\meg 1} ) \cdot (\partial_r \cdot \textbf{n}_{\si} ) dg_{\si} + $$ $$ + \int_{\cala_{0}^{\tau}} \partial_r  \left( \dfrac{2M\sqrt{D}}{r} \slashed{\Delta} \psi_{\meg 1} \right) \cdot (\partial_{r}\psi_{\meg 1} ) dg_{\cala} + \int_{\cala_{0}^{\tau}} \dfrac{2M\sqrt{D}}{r} (\partial_{r}\psi_{\meg 1} ) \cdot (\slashed{\Delta} \psi_{\meg 1} ) \cdot \frac{2}{r} dg_{\cala} , $$
by an application of Stokes' Theorem.  

For the last two spacetime terms we have that (after also using Stokes' Theorem on the sphere)
\begin{equation}\label{h81}
 \int_{\cala_{0}^{\tau}} \partial_r \left( \dfrac{2M\sqrt{D}}{r} \slashed{\Delta} \psi_{\meg 1} \right) \cdot (\partial_{r}\psi_{\meg 1} ) dg_{\cala} + \int_{\cala_{0}^{\tau}} \dfrac{2M\sqrt{D}}{r} (\partial_{r}\psi_{\meg 1} ) \cdot (\slashed{\Delta} \psi_{\meg 1} ) \cdot \frac{2}{r}dg_{\cala} = 
\end{equation} 
$$ = \int_{\cala_{0}^{\tau}} \dfrac{2M\sqrt{D}}{r} (\slashed{\Delta} \partial_r \psi_{\meg 1} - [\slashed{\Delta} , \partial_r ] \psi_{\meg 1} ) \cdot (\partial_{r}\psi_{\meg 1} ) dg_{\cala} + $$ $$ +  \int_{\cala_{0}^{\tau}}\left( \partial_r \left(  \dfrac{2M\sqrt{D}}{r} \right) + \frac{2}{r} \cdot \dfrac{2M\sqrt{D}}{r} \right) (\partial_{r}\psi_{\meg 1} ) \cdot (\slashed{\Delta} \psi_{\meg 1} ) dg_{\cala} = $$ $$ = \int_{\cala_{0}^{\tau}} \dfrac{2M\sqrt{D}}{r} (\partial_{r}\psi_{\meg 1} ) \cdot (\slashed{\Delta} \partial_r \psi_{\meg 1} ) dg_{\cala}  + \int_{\cala_{0}^{\tau}} \partial_r \left( \dfrac{2M\sqrt{D}}{r} \right) (\partial_{r}\psi_{\meg 1} ) \cdot (\slashed{\Delta} \psi_{\meg 1} ) dg_{\cala} = $$
  $$ = - \int_{\cala_{0}^{\tau}} \dfrac{2M\sqrt{D}}{r} |\slashed{\nabla} \partial_r \psi_{\meg 1} |^2 dg_{\cala} - \int_{\cala_{0}^{\tau}} \partial_r \left( \dfrac{2M\sqrt{D}}{r} \right) \langle \slashed{\nabla} \psi_{\meg 1} , \slashed{\nabla} \partial_{r}\psi_{\meg 1} \rangle dg_{\cala} , $$
where we used that $[ \slashed{\Delta} , \partial_r ] \psi_{\meg 1} = \dfrac{2}{r} \slashed{\Delta} \psi_{\meg 1}.$  
  
Since 
$$ \partial_r\left( \dfrac{2M\sqrt{D}}{r} \right) = - \dfrac{2M\sqrt{D}}{r^2} + \frac{2M^2}{r^3} \simeq \frac{2}{M} \mbox{ in $\cala$,} $$
the second term on the right hand side of \eqref{h81} can be treated as follows:
$$ - \int_{\cala_{0}^{\tau}} \partial_r \left( \dfrac{2M\sqrt{D}}{r} \right) \langle \slashed{\nabla} \psi_{\meg 1} , \slashed{\nabla} \partial_{r}\psi_{\meg 1} \rangle dg_{\cala} \mik \beta_1 \int_{\cala_0^{\tau}} |\slashed{\nabla} \partial_{r}\psi_{\meg 1} |^2 dg_{\cala} + \frac{1}{\beta_1} \int_{\cala_0^{\tau}} |\slashed{\nabla} \psi_{\meg 1} |^2 dg_{\cala} , $$
where $\beta_1 = \beta_1 (M)$ is chosen to be small enough so that the first term can be absorbed from the right hand (from $H_{3}$), while the second term can be bounded by the Morawetz estimate.

The first term on the right hand side of \eqref{h81} can be absorbed as well from $H_3$ since $\dfrac{2M\sqrt{D}}{r} \ll H_3$ close to the horizon.

Finally for the second term on the right hand side of \eqref{h8} we have after integrating by parts on the sphere that
$$ \int_{\si_{\tau}} \dfrac{2M\sqrt{D}}{r} (\partial_{r}\psi_{\meg 1} ) \cdot (\slashed{\Delta} \psi_{\meg 1} ) \cdot (\partial_r \cdot \textbf{n}_{\si} ) dg_{\si} = $$ $$ = -\int_{\si_{\tau}} \dfrac{2M\sqrt{D}}{r} \langle \slashed{\nabla} \partial_{r}\psi_{\meg 1} ) ,  \slashed{\nabla} \psi_{\meg 1}  \rangle \cdot (\partial_r \cdot \textbf{n}_{\si} ) dg_{\si} \mik $$ $$ \mik \int_{\si_{\tau}} \dfrac{4M^2 D}{r^2} |\slashed{\nabla} \partial_{r}\psi_{\meg 1} |^2 dg_{\si} + \int_{\si_{\tau}} |\slashed{\nabla} \psi_{\meg 1} |^2 dg_{\si} , $$
where we can absorb the first term from the energy identity for $J^{L_{P}}[\partial_{r}\psi_{\geq 1}]$ (since $\frac{4M^{2}D}{r^{2}}\rightarrow 0$ as $r\rightarrow M$), and the second one is bounded by the $T$-flux of $\psi_{\meg 1}$. We can treat the first term of the right hand side of \eqref{h8} similarly.

Note that the $H_8$ term introduced only a $\beta$ loss (for $\beta$ very small) from the $H_3$ term.

\textbf{$H_9$}: We have that
$$ \int_{\cala_{0}^{\tau}} H_9 (\partial_{v}\partial_{r}\psi_{\meg 1} ) \cdot (\partial_{r}\partial_{r}\psi_{\meg 1} ) dg_{\cala} \mik $$ $$ \mik \int_{\cala_{0}^{\tau}}\sqrt{D} \left(\frac{\sqrt{D}}{\sigma} + \frac{2M}{r^2} + \frac{2}{r} \right) (\partial_{v}\partial_{r}\psi_{\meg 1} ) \cdot (\partial_{r}\partial_{r}\psi_{\meg 1} ) dg_{\cala} \mik $$ $$ \mik \beta \int_{\cala_{0}^{\tau}} D  (\partial_{r}\partial_{r}\psi_{\meg 1} )^2 dg_{\cala} + \frac{1}{\beta} \int_{\cala_{0}^{\tau}} (\partial_{v}\partial_{r}\psi_{\meg 1} )^2 dg_{\cala} , $$
where $\beta = \beta (M)$ is chosen to be very small so that the first term of the last inequality can be absorbed by $H_2$, but also to satisfy $\frac{1}{\beta} \ll \frac{1}{\sigma}$ (which is possible by choosing $\sigma$ to be extremely small from the beginning) so that the second term can be absorbed by $H_1$ as well.

Note that the $H_9$ term introduced only a $\beta$ loss (for $\beta$ very small) from the $H_1$ and $H_2$ terms.

\textbf{$H_{10}$}: We have that
\begin{equation}\label{h10}
 \int_{\cala_{0}^{\tau}} H_{10} (\partial_{r}\partial_{r}\psi_{\meg 1} ) \cdot (\partial_{r}\psi_{\meg 1} ) dg_{\cala} \mik 
 \end{equation}
 $$ \mik \frac{1}{2\beta} \int_{\cala_{0}^{\tau}} (H_{10})^2 \frac{M^2}{2} (\partial_{r}\partial_{r}\psi_{\meg 1} )^2 dg_{\cala} + \frac{\beta}{2} \int_{\cala_{0}^{\tau}} \frac{2}{M^2} (\partial_{r}\psi_{\meg 1} )^2 dg_{\cala} , $$
where we just applied the Cauchy--Schwarz inequality, for some $\beta$ that will be chosen later.

We deal first with the second term for which we apply Poincar\'{e}'s inequality:
$$\frac{\beta}{2} \int_{\cala_{0}^{\tau}} \frac{2}{M^2} (\partial_{r}\psi_{\meg 1} )^2 dg_{\cala} \mik \frac{\beta}{2} \frac{(M+\beta' )^2}{M^2} \int_{\cala_{0}^{\tau}}  |\slashed{\nabla} \partial_{r}\psi_{\meg 1} |^2 dg_{\cala} ,$$
where we denoted the $r_0$ given in the definition of $\cala$ by $r_0 = M + \beta'$ for some very small $\beta' > 0$. Now we note that for an appropriate choice of $\beta$ close to 1 we can have that:
$$ H_3 = \frac{M^2}{2r^2} > \frac{\beta}{2} \frac{(M+\beta' )^2}{M^2} \Rightarrow H_3 - \frac{\beta}{2} \frac{(M+\beta' )^2}{M^2} > c > 0 , $$
so in the end the second term of the right hand side of \eqref{h10} can be absorbed by the $H_3$ term. Note that this is possible also because of the fact that from all the previous terms $H$ terms that we examined, we only had $\beta''$ loss in $H_3$ for $\beta'' > 0$ being extremely small.

We now look at the first term of the right hand side of \eqref{h10}. We have that:
$$ (H_{10} )^2 = M^2 D (R' )^2 = M^2 D \left( \frac{6M^2}{r^4} - \frac{4M}{r^3} + \frac{2M \sqrt{D}}{r^3} - \frac{2M D}{r^2} \right)^2 . $$
In $\cala$, the first two terms of $R'$ are the important ones, since they are much bigger than the last two. We recall the following precise estimate:
$$ \frac{6M^2}{r^4} - \frac{4M}{r^3} + \frac{2M \sqrt{D}}{r^3} - \frac{2M D}{r^2} \simeq \frac{2M^2}{r^4} \mbox{ in $\cala$}, $$
that we also used in the estimate for $H_4$. 

We would like to show that:
$$ H_2 = \frac{5M^2 D}{2r^2} + \frac{MD^{3/2}}{r} > \frac{1}{2\beta}\cdot\frac{M^2}{2} \cdot \frac{4M^6 D}{r^8} = \frac{M^8 D}{\beta r^8} . $$
Indeed, by our choice of $\beta$ (which as we mentioned it is chosen to be close to 1), we can have that:
$$ \frac{5M^2}{2r^2} > \frac{M^8 D}{\beta r^8} , $$
and this proves the required estimate (note that the term $\frac{MD^{3/2}}{r}$ is of lower order in $(r-M)$ and hence cannot be used in the proof of the estimates).

\end{proof}

\begin{remark}It is important to emphasize that the exact value of the coefficients plays a crucial role in bounding the above terms. This reveals a new structure of degenerate horizons.
\end{remark}

\subsection{Finishing the proofs of the theorems}
\label{sec:FinishingTheProofsOfTheTheorems}

We now have all the tools necessary to complete the proofs of Theorems \ref{theo1}, \ref{theo2} and \ref{theo3}. 

\begin{proof}[Proof of Theorem \ref{theo1}]
Clearly, in view of the degenerate Morawetz estimate \eqref{dmora} and the Hardy inequality \eqref{1hardy} (that controls the zeroth order term) it suffices to show that 
\begin{equation}
\begin{split}
\int_{\mathcal{A}}(\partial_{r}\psi)^{2}\, dg_{\mathcal{A}}\leq  & \ C \cdot D^{w}_{\Sigma_{0}}[\psi], 
\end{split}
\label{e1theo2}
\end{equation}
where the norm $D^{w}_{\Sigma_{0}}[\psi]$ is defined by \eqref{norminitial1}. We decompose 
\[ \psi=\psi_{0}+\psi_{\geq 1}\]
as in Section \ref{sec:PoincarEInequality}. In view of the orthogonality of $\psi_{0}$ and $\psi_{\geq 1}$ in $L^{2}\left(\mathbb{S}^{2}\right)$ we have 
\begin{equation}
\int_{\mathcal{A}}(\partial_{r}\psi)^{2}\, dg_{\mathcal{A}}=\int_{\mathcal{A}}(\partial_{r}\psi_{0})^{2}\, dg_{\mathcal{A}}+\int_{\mathcal{A}}(\partial_{r}\psi_{\geq 1})^{2}\, dg_{\mathcal{A}}.
\label{decoproof}
\end{equation}
For the spherically symmetric term we apply Proposition \ref{prop1l0}. For the term $\psi_{\geq 1}$ we apply the Hardy inequality \eqref{1hardy} for $p=0$ combined with the degenerate Morawetz estimate to get
\begin{equation}
\int_{\mathcal{A}}(\partial_{r}\psi_{\geq 1})^{2}\, dg_{\mathcal{A}} \leq C \int_{\mathcal{A}}(r-M)^{2}\cdot \big(\partial_{r}\partial_{r}\psi_{\geq 1}\big)^{2}\, dg_{\mathcal{A}}+C\int_{\Sigma_{0}}\Big(J^{T}[\psi_{\geq 1}]\cdot \textbf{n}_{\Sigma_{0}}\Big)\, dg_{\Sigma_{0}}.
\label{l1proof1}
\end{equation}
The first integral on the right hand side can be bounded by Proposition \ref{2penes2}. By adding the two estimates and using \eqref{decoproof} we obtain the desired results.

\end{proof}

\begin{proof}[Proof of Theorem \ref{theo2}]
Let $\psi$ be a spherically symmetric solution to the wave equation with smooth compactly supported initial data such that the conserved charge 
\begin{equation}
H[\psi]=\int_{\Sigma_{0}\cap\mathcal{H}}\left(\partial_{r}\psi+\frac{1}{M}\psi\right)\, d\omega \ =\,1.
\label{proof1theo2}
\end{equation}
Then we clearly have that
\begin{equation}
\int_{\Sigma_{0}\cap\left\{r\leq R\right\}}\left[\frac{1}{(r-M)}\cdot \left(\partial_{r}\psi+\frac{1}{r}\psi\right)^{2}\right]dr\,d\omega= \infty. 
\label{p2theo22}
\end{equation}
On the other hand, we have established that for the spherically symmetric solution $\psi$ the exact \textbf{identity} \eqref{l0f2} holds, which can be re-written as 
\begin{equation}
\begin{split}
 \int_{\mathcal{A}_{r_{0}}}&\left[\left(\frac{3}{2}-\frac{(r-M)}{r}\right)\cdot (\partial_{r}\psi_{0})^{2}+\left(\frac{r-M}{r^{3}}\right)\cdot \psi_{0}^{2}\right]dr\,dv\,d\omega \\
&+\int_{H_{r_{0}}}\left[\frac{1}{2r^{2}}\cdot\psi_{0}^{2}\right]dv\,d\omega +\int_{B_{R}}\Big[K_{R}[\psi_{0}]\Big]dv\,d\omega\\
=& \int_{\Sigma_{0}\cap\left\{r_{0}\leq r\leq R\right\}}\left[\frac{r^{2}}{(r-M)}\cdot\left(\partial_{r}\psi_{0}+\frac{1}{r}\psi_{0}\right)^{2}\right]dr\,d\omega 
\end{split}
\label{l0f2theoproof}
\end{equation}
where $K_{R}[\psi]$ is given by \eqref{krpsi}.
In view of \eqref{p2theo22}, the right hand side of \eqref{l0f2theoproof}  tends to infinity as $r_{0}\rightarrow M$. Since the initial data of $\psi$ are assumed to be smooth and compacly supported,  we have that the $T$-flux of $\psi$ through $\Sigma_{0}$ is finite. Hence, the integrals in \eqref{l0f2theoproof} over the hypersurface $B_{R}$  is uniformly (in $r_{0}$) bounded using the averaging principle and the degenerate Morawetz estimate \eqref{dmora}. On the other hand, if we assume that
\begin{equation}
\int_{\mathcal{A}}(\partial_{r}\psi)^{2}\, dg_{\mathcal{A}} <\infty
\label{contra1proof2}
\end{equation}
then the integral in \eqref{l0f2theoproof} over the hypersurface $H_{r_{0}}$  is uniformly (in $r_{0}$) bounded using the Hardy inequality \eqref{2hardy} and  the zeroth order term in the integral over $\mathcal{A}_{r_{0}}$ is uniformly (in $r_{0}$) bounded using the Hardy inequality \eqref{1hardy}. 
Hence, if we assume \eqref{contra1proof2} then the limit as $r_{0}\rightarrow M$ of left hand side of \eqref{l0f2theoproof} is finite, which, since $\mathcal{A}_{r_{0}}$ tends to $\mathcal{A}$ as $r_{0}\rightarrow M$, contradicts \eqref{contra1proof2}!   This completes the proof of the theorem. 
\end{proof}

\begin{proof}[Proof of Theorem \ref{theo3}] We consider initial data for $\psi$ with vanishing  conserved charge $H[\psi]=0$ but which are singular in the sense that  \eqref{optimalre2} holds. Clearly such the data are not $C^{2}$.  

 The proof in this case mimics that of the proof of Theorem \ref{theo2}. Indeed we use the identity \eqref{l0f2theoproof} for the spherical mean of $\psi$ and argue by contradiction. Assuming that \eqref{contra1proof2} holds we obtain that left hand side of \eqref{l0f2theoproof} is uniformly bounded in $r_{0}$. This however contradicts the fact that the right hand side of \eqref{l0f2theoproof} blows up as $r_{0}\rightarrow M$.

\end{proof}

\subsection{Remarks about the singular multiplier $S$ }
\label{sec:Discussion}

The proof of Theorem \ref{theo1} heavily relies on the use of the singular vector field 
\[S=-\frac{1}{r-r_{hor}}\cdot \partial_{r}\]
as a multiplier vector field. Here $r_{hor}$ is the radius of the event horizon and hence this vector field is singular on the event horizon. 

As was noted in Section \ref{sec:TheEstimateForTheSphericalMean}, the most critical term in the analysis is the integral $I_{4}$ given by \eqref{i4}. Note that the boundeness of the limit as $r_{0}\rightarrow r_{hor}$ (where the event horizon is located at $r=r_{hor}$) of this integral in the extremal case is a new feature of the geometry of degenerate horizons. 
In other words, the product $\partial_{r}\psi\cdot\partial_{v}\psi$ oscillates in time faster and faster as we approach the event horizon forcing thus the singular integral $I_{4}$ to have a finite limit. Such an oscillation cease to hold in the sub-extremal case.  Indeed, we will next show that for generic solutions to the wave equation on sub-extremal black holes we have 
\begin{equation}
\lim_{r_{0}\rightarrow r_{hor}}\int_{\mathcal{A}_{r_{0}}}\frac{1}{(r-r_{hor})}\partial_{r}\psi \cdot \partial_{v}\psi\, dg_{\mathcal{A}_{r_{0}}} =\infty.
\label{dis1}
\end{equation} 
Let us consider initial data in the class $\mathcal{C}_{data}$ given by
\[ \mathcal{C}_{data}=\Big\{\text{smooth data supported on }\Sigma_{0}\cap\left\{r_{hor}<R_{1}\leq r\leq R_{2}\right\}    \Big\},\]
for some constants $R_{1},R_{2}>r_{hor}$. 
 Note that in the sub-extremal case we \textit{schematically} have
\begin{equation}
\Box_{g}\psi= (r-r_{hor})\cdot \partial_{r}\partial_{r}\psi+\partial_{r}\partial_{v}\psi+\partial_{v}\psi+\partial_{r}\psi+\lapp\psi=0
\label{schema1}
\end{equation}
and
\begin{equation}
dr =(r-r_{hor})\cdot dv \text{ : on }H_{r_{0}}. 
\label{schema2}
\end{equation}
We next restrict our attention to spherically symmetric solutions to the wave equation on subextremal backgrounds. Let us \textbf{assume that for \underline{all} such solutions with initial data in the class $\mathcal{C}_{data}$ the following integral is in fact bounded}:
\begin{equation}
\lim_{r_{0}\rightarrow r_{hor}}\int_{\mathcal{A}_{r_{0}}}\frac{1}{(r-r_{hor})}\partial_{r}\psi \cdot \partial_{v}\psi\, dg_{\mathcal{A}_{r_{0}}} <\infty,
\label{discontra}
\end{equation}
Then, we clearly have that if $\psi$ has data in the class $\mathcal{C}_{data}$  then \[\varphi=\partial_{v}\psi\] also has data in $\mathcal{C}_{data}$.  Hence, 
\begin{equation}
\lim_{r_{0}\rightarrow r_{hor}}\int_{\mathcal{A}_{r_{0}}}\frac{1}{(r-r_{hor})}\partial_{r}\varphi \cdot \partial_{v}\varphi\, dg_{\mathcal{A}_{r_{0}}} <\infty. 
\label{dis2}
\end{equation}
By integrating with respect to $\partial_{v}$ we obtain
\begin{equation}
\begin{split}
\lim_{r_{0}\rightarrow r_{hor}}&\int_{\mathcal{A}_{r_{0}}}\frac{1}{(r-r_{hor})}\partial_{r}\varphi \cdot \partial_{v}\varphi\, dr\,dv\,d\omega  \\=&\int_{H_{r_{0}}}\frac{1}{r-r_{hor}} \cdot\varphi \cdot \partial_{r}\varphi \, dr\,d\omega -\int_{\Sigma_{0}}\frac{1}{r-r_{hor}} \cdot\varphi \cdot \partial_{r}\varphi \, dr\,d\omega\\&-\int_{\mathcal{A}_{r_{0}}}\frac{1}{(r-r_{hor})}\partial_{r}\partial_{v}\varphi \cdot \varphi\, dr\,dv\,d\omega.
\label{dis3}
\end{split}
\end{equation}
In view of \eqref{schema2} we have
\begin{equation}
\begin{split}
\int_{H_{r_{0}}}\frac{1}{r-r_{hor}} \cdot\varphi \cdot \partial_{r}\varphi \, dr\,d\omega =&\int_{H_{r_{0}}} \varphi \cdot \partial_{r}\varphi \, dv\,d\omega \\&
\leq \int_{H_{r_{0}}} \Big[\varphi^{2} + (\partial_{r}\varphi)^{2}\Big] dv\,d\omega\\
&\leq C\int_{\Sigma_{0}}\Big(\sum_{k=1,2}J^{N}[N^{k}\varphi]\cdot \textbf{n}_{\Sigma_{0}}\Big)dg_{\Sigma_{0}}<\infty,
\label{dis4}
\end{split}
\end{equation}
where in the last step we used the redshift estimate of Dafermos and Rodnianski \cite{lecturesMD}. Note that the bound is uniform in $r_{0}$.

The integral over $\Sigma_{0}$ in \eqref{dis3} is finite since the integrand quantity depends only on the initial data of $\varphi$ and by assumption these data are supported away from the horizon. Furthermore, in view of \eqref{schema1}, we schematically obtain
\begin{equation}
\begin{split}
\int_{\mathcal{A}_{r_{0}}}&\frac{1}{(r-r_{hor})}\partial_{r}\partial_{v}\varphi \cdot \varphi\, dr\,dv\,d\omega \\=&\int_{\mathcal{A}_{r_{0}}}\frac{1}{(r-r_{hor})}\Big[(r-r_{hor})\cdot\partial_{r}\partial_{r}\varphi+\partial_{v}\varphi+\partial_{r}\varphi\Big]\cdot \varphi\, dr\,dv\,d\omega
\\=&\int_{\mathcal{A}_{r_{0}}}\left[\partial_{r}\partial_{r}\varphi+   \frac{1}{(r-r_{hor})}\partial_{v}\varphi+\frac{1}{(r-r_{hor})}\partial_{r}\varphi\right]\cdot \varphi\, dr\,dv\,d\omega
\label{dis5}
\end{split}
\end{equation}
Similarly as above we have
\begin{equation}
\begin{split}
\int_{\mathcal{A}_{r_{0}}}&\partial_{r}\partial_{r}\varphi \cdot \varphi\, dr\,dv\,d\omega\\ \leq &\int_{\mathcal{A}_{r_{0}}}  \Big[\varphi^{2} + (\partial_{r}\partial_{r}\varphi)^{2}\Big] dv\,d\omega\\
\leq &   \ C\int_{\Sigma_{0}}\Big(\sum_{k=1,2}J^{N}[N^{k}\varphi]\cdot \textbf{n}_{\Sigma_{0}}\Big)dg_{\Sigma_{0}}<\infty,
\label{dis6}
\end{split}
\end{equation}
where we used once again the redshift estimate. Furthermore, in view of \eqref{schema2}
\begin{equation}
\begin{split}
\int_{\mathcal{A}_{r_{0}}}&\frac{1}{(r-r_{hor})}\partial_{v}(\varphi^{2})\, dr\,dv\,d\omega\\ =&
\int_{H_{r_{0}}}\frac{1}{(r-r_{hor})}\varphi^{2}\, dr\,d\omega-\int_{\Sigma_{0}}\frac{1}{(r-r_{hor})}\varphi^{2}\, dr\,d\omega\\ =&
\int_{H_{r_{0}}}\varphi^{2}\, dv\,d\omega-\int_{\Sigma_{0}}\frac{1}{(r-r_{hor})}\varphi^{2}\, dr\,d\omega\\ \leq &\ 
C\int_{\Sigma_{0}}\Big(J^{N}[\varphi]\cdot \textbf{n}_{\Sigma_{0}}\Big)dg_{\Sigma_{0}}-\int_{\Sigma_{0}}\frac{1}{(r-r_{hor})}\varphi^{2}\, dr\,d\omega <\infty
\label{dis7}
\end{split}
\end{equation}
since $\varphi$ is initially supported away from the horizon, where to bound the first integral we have used the Hardy inequality \eqref{2hardy} and the non-degenerate Morawetz estimate for sub-extremal black holes. Regarding the remaining term we schematically have the following
\begin{equation}
\begin{split}
\int_{\mathcal{A}_{r_{0}}}&\frac{1}{(r-r_{hor})}\cdot\partial_{r}\varphi\cdot \varphi\, dr\,dv\,d\omega\\ =&
\int_{\mathcal{A}_{r_{0}}}\frac{1}{(r-r_{hor})}\cdot\partial_{r}\partial_{v}\psi\cdot \partial_{v}\psi\, dr\,dv\,d\omega\\ =&
\int_{\mathcal{A}_{r_{0}}}\frac{1}{(r-r_{hor})}\cdot\Big[(r-r_{hor})\cdot\partial_{r}\partial_{r}\psi+\partial_{v}\psi+\partial_{r}\psi\Big]\cdot \partial_{v}\psi\, dr\,dv\,d\omega\\ =&
\int_{\mathcal{A}_{r_{0}}}\left[\partial_{r}\partial_{r}\psi+\frac{1}{(r-r_{hor})}\cdot\partial_{v}\psi+\frac{1}{(r-r_{hor})}\cdot\partial_{r}\psi\right]\cdot \partial_{v}\psi\, dr\,dv\,d\omega.
\label{dis8}
\end{split}
\end{equation}
As before, by Cauchy--Schwarz and the redshift estimate we have 
\[\int_{\mathcal{A}_{r_{0}}}\partial_{r}\partial_{r}\psi\cdot \partial_{v}\psi\, dr\, dv\, d\omega \leq\  C\int_{\Sigma_{0}}\Big(\sum_{k=1,2}J^{N}[N^{k}\varphi]\cdot \textbf{n}_{\Sigma_{0}}\Big)dg_{\Sigma_{0}}<\infty.  \]
By our assumption \eqref{discontra} we also have
\begin{equation}
\lim_{r_{0}\rightarrow r_{hor}}\int_{\mathcal{A}_{r_{0}}}\frac{1}{(r-r_{hor})}\partial_{r}\psi \cdot \partial_{v}\psi\, dr\,dv\,d\omega\ < \ \infty.
\label{dis10}
\end{equation}
Hence, in view of \eqref{dis2}, \eqref{dis3}, \eqref{dis4}, \eqref{dis5}, \eqref{dis6}, \eqref{dis7}, \eqref{dis8} and \eqref{dis10},  finally obtain that 
\begin{equation}
\lim_{r_{0}\rightarrow r_{hor}}\int_{\mathcal{A}_{r_{0}}}\frac{1}{(r-r_{hor})}\cdot (\partial_{v}\psi)^{2}\, dr\,dv\,d\omega\ < \ \infty.
\label{disfcontra}
\end{equation}
However, the integrand quantity is non-negative definite and moreover for generic $\psi$ we have that there is a $v_{1}>0$ such that \begin{equation}
(\partial_{v}\psi)(v_{1},r=r_{hor})\neq 0.
\label{nongenonh}
\end{equation}
One can explicitly construct such a solution to the wave equation by imposing appropriate characteristic initial data. 
  Since the function $\frac{1}{(r-r_{hor})}$ is not in $L^{1}_{\text{loc}}$ in $\mathcal{A}$, the  condition \eqref{nongenonh}  immediately forces the limit on the left hand side of \eqref{disfcontra} to be infinite, which is a contradiction.

\section{Higher order estimates and stable trapping}
\label{sec:HigherOrderEstimatesAndStableTrapping}

In this section we prove Theorem \ref{theo4}. We show that for generic smooth initial data there is no non-degenerate higher order Morawetz estimate for solutions arising from smooth initial data on a Cauchy hypersurface $\Sigma_{0}$ supported in a compact region $\left\{M<\tilde{R_{1}}\leq r\leq r\leq \tilde{R_{2}}\right\}$ away from $\Sigma_{0}\cap \mathcal{H}$, where $\mathcal{H}=\left\{r=M\right\}$
denotes the event horizon.

\begin{proof}[Proof of Theorem \ref{theo4}]

 First divide $\Sigma_{0}$ in the following regions:
$$ \si_0 \doteq  \si_0^1 \cup \si_0^2\cup \si_{0}^3  $$
where
$$ \si_0^1 = \si_0 \cap \{ M \mik r \mik \tilde{R_{1}} \}, \quad  \si_0^2 = \si_0 \cap \{ \tilde{R_{1}} \mik r \mik \tilde{R_{2}} \}, \quad \si_0^3 = \si_0 \cap \{ r \meg \tilde{R_{2}} \}, $$
where $M < \tilde{R_{1}} < \tilde{R_{2}} < \infty$.
\begin{figure}[H]
\begin{center}
\includegraphics[width=3.5in]{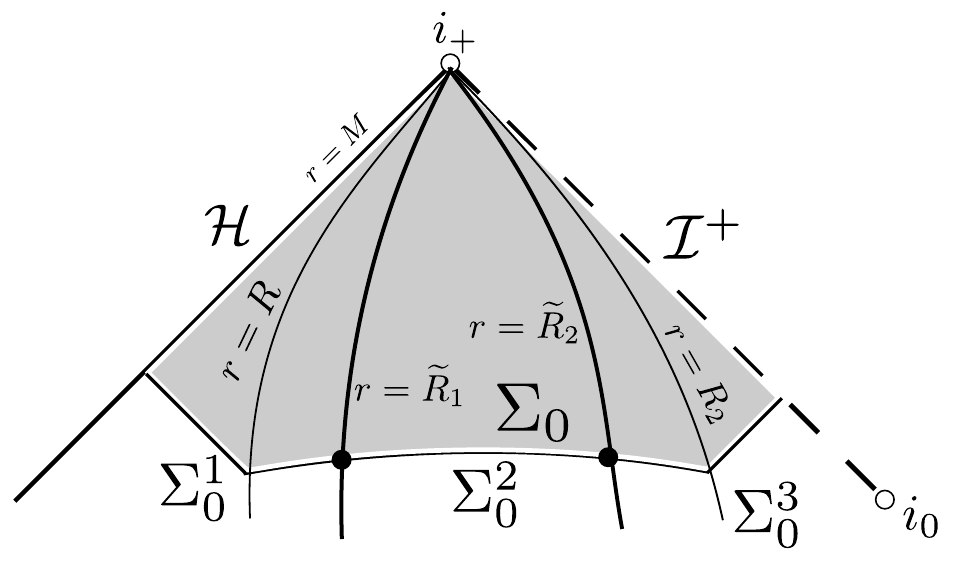}
\end{center}
\end{figure}
We consider the following data on $\si_0$:
\begin{equation}
\begin{split}
 \left. \phi \right|_{ \si_0^1} = 1 , & \quad \quad \left. \phi \right|_{\si_0^3} = 0 , \\
 \left. \partial_v \phi \right|_{ \si_0^1} = 0 ,& \quad \quad \left. \partial_v \phi \right|_{\si_0^3} = 0 ,
\end{split}
\label{phiconstruction}
\end{equation}
and $\left. \phi \right|_{\si_0^2}$ and $\left. \partial_v \phi \right|_{\si_0^2}$ are smooth and spherically symmetric. There data give rise to a spherically symmetric solution $\phi$ that is smooth in the domain of outer communications. Note also that 
\begin{equation}
H[\phi]=\frac{1}{M}\neq 0.
\label{consphi}
\end{equation}
We  now consider:
\begin{equation}
\psi:=\partial_{v}\phi
\label{psiconstruction}
\end{equation}
which is also a solution to the wave equation by the fact that $[ \partial_v , \Box_g ] = 0$. By Proposition 1 of \cite{aretakis2012} we have that the initial data of $\psi$ on $\Sigma_{0}$ are smooth and supported in a compact region $\left\{M<\tilde{R_{1}}\leq r\leq \tilde{R_{2}}\right\}$. Hence, it follows  that $H [\psi ] = 0$.  More specifically we have that:
$$ \left. \psi \right|_{\si_0^1} = 0 , \quad \quad \left. \psi \right|_{\si_0^3} = 0 , $$
$$ \left. \partial_v \psi \right|_{\si_0^1} = 0 , \quad \quad \left. \partial_v \psi \right|_{\si_0^3} = 0 ,$$
and $\left. \psi \right|_{\si_0^2}$ and $\left. \partial_v \psi \right|_{\si_0^2}$ are smooth and spherically symmetric.

We will next show that for  $\psi$ given by \eqref{psiconstruction} we have 
\begin{equation}
\int_{\mathcal{A}}\big(\partial_{r}\partial_{r}\psi\big)^{2}\, dg_{\mathcal{A}}=\infty.
\label{whattoshow}
\end{equation}
In view of the wave equation for $\phi$ we have
\begin{equation}
-2\partial_{r}\partial_{v}\phi=D\cdot\partial_{r}\partial_{r}\phi+\frac{2}{r}\partial_{v}\phi+R\cdot\partial_{r}\phi
\label{wephi}
\end{equation}
and hence commuting with the vector field $\partial_{r}$ we obtain
\begin{equation}
-2\partial_{r}\partial_{r}\partial_{v}\phi=D\cdot\partial_{r}\partial_{r}\partial_{r}\phi+\Big(\partial_{r}D+R\Big)\cdot \partial_{r}\partial_{r}\phi + \partial_{r}R\cdot \partial_{r}\phi+\partial_{r}\left(\frac{2}{r}\partial_{v}\phi\right).
\label{wephi2}
\end{equation}
By virtue of \eqref{defsl0} we have
\[r^{2}\cdot\big(\partial_{r}D+R\big)=2M\frac{(r-M)}{r}+2(r-M)= 4(r-M)-\frac{2}{r}\cdot(r-M)^{2}\]
and
\[r^{2}\cdot\partial_{r}R=2-\frac{4}{r}\cdot(r-M). \]
Hence, if we define
\begin{equation}
\mathcal{E}_{1}[\phi]:=
2r\cdot\partial_{r}\partial_{v}\phi-2\partial_{v}\phi-\frac{2}{r}(r-M)^{2}\cdot \partial_{r}\partial_{r}\phi-\frac{4}{r}(r-M)\cdot \partial_{r}\phi.
\label{errorhot1}
\end{equation}
then we obtain
\begin{equation}
4r^{2}\cdot |\partial_{r}\partial_{r}\partial_{v}\phi|=(r-M)^{2}\cdot\partial_{r}\partial_{r}\partial_{r}\phi+4(r-M)\cdot \partial_{r}\partial_{r}\phi + 2 \cdot \partial_{r}\phi+\mathcal{E}_{1}[\phi]
\label{wephi2}
\end{equation}
and, therefore,
\begin{equation}
4r^{2}\cdot |\partial_{r}\partial_{r}\partial_{v}\phi|=\partial_{r}\partial_{r}\Big((r-M)^{2}\cdot\partial_{r}\phi\Big)+\mathcal{E}_{1}[\phi].
\label{wephi3}
\end{equation}
By the Cauchy--Schwarz inequality we obtain
\begin{equation}
16r^{4}\cdot \Big(|\partial_{r}\partial_{r}\partial_{v}\phi|\Big)^{2}\geq (1-\epsilon)\cdot\Bigg(\partial_{r}\partial_{r}\Big((r-M)^{2}\cdot\partial_{r}\phi\Big)\Bigg)^{2}+\left(1-\frac{1}{\epsilon}\right)\cdot \Big(\mathcal{E}_{1}[\phi]\Big)^{2}
\label{wephi3}
\end{equation}
for sufficiently small $\frac{1}{2}>\epsilon>0$. Therefore, after choosing such $\epsilon$ and using that $r$ is bounded in $\mathcal{A}$, there exists a positive constant $C$ that depends only on $M$ such that 
\begin{equation}
\begin{split}
&\int_{\mathcal{A}}\Bigg(\partial_{r}\partial_{r}\Big((r-M)^{2}\cdot\partial_{r}\phi\Big)\Bigg)^{2}\, dg_{\mathcal{A}}\leq  \\
&\ \ \ \ \ \ \ \ \ \ \ \ \ C\int_{\mathcal{A}}r^{4}\cdot \Big(|\partial_{r}\partial_{r}\partial_{v}\phi|\Big)^{2}\, dg_{\mathcal{A}}\  +C\int_{\mathcal{A}}\Big(\mathcal{E}_{1}[\phi]\Big)^{2}\, dg_{\mathcal{A}}
\label{wephi4}
\end{split}
\end{equation}
We will next show that 
\begin{equation}
\int_{\mathcal{A}}\Big(\mathcal{E}_{1}[\phi]\Big)^{2}\, dg_{\mathcal{A}} <\infty.
\label{finite1}
\end{equation}
First we observe that in view of the wave equation \eqref{wel0} for $\phi$
we can write
\[-\frac{2}{r}(r-M)^{2}\cdot\partial_{r}\partial_{r}\phi=4r\cdot\partial_{r}\partial_{v}\phi+4\partial_{v}\phi+\frac{4}{r}(r-M)\cdot\partial_{r}\phi \]
and hence  $\mathcal{E}_{1}[\phi]$ becomes
\begin{equation}
\mathcal{E}_{1}[\phi]=6r\cdot\partial_{r}\partial_{v}\phi+2\partial_{v}\phi.
\label{enewexp}
\end{equation}
Therefore,
\begin{equation}
\begin{split}
\int_{\mathcal{A}}\Big(\mathcal{E}_{1}[\phi]\Big)^{2}\, dg_{\mathcal{A}}\leq &\ C\int_{\mathcal{A}}\Big((\partial_{v}\phi)^{2}+(\partial_{r}\partial_{v}\phi)\Big)^{2}\, dg_{\mathcal{A}}
\\ \leq & \ C\int_{\mathcal{A}}\Big((\partial_{v}\phi)^{2}+(\partial_{r}\psi)\Big)^{2}\, dg_{\mathcal{A}}.
\label{efinite1}
\end{split}
\end{equation}
In view of the degenerate Morawetz theorem established in \cite{aretakis2} we have 
\begin{equation}
\int_{\mathcal{A}}(\partial_{v}\phi)^{2}\, dg_{\mathcal{A}} \leq \ C\int_{\Sigma_{0}}\left(J^{N}[\phi]\cdot \textbf{n}_{\Sigma_{0}} \right)dg_{\Sigma_{0}}<\infty
\label{eee1}
\end{equation}
since $\phi$ is a smooth on $\Sigma_{0}$. Furthermore, in view of 
the Proposition \ref{prop1l0} we have 
\begin{equation}
\begin{split}
\int_{\mathcal{A}}(\partial_{r}\psi)^{2}\, dg_{\mathcal{A}}\leq  & \ C \cdot D_{\Sigma_{0}}[\psi] <\infty
\end{split}
\label{eee2}
\end{equation}
since \[D_{\Sigma_{0}}[\psi] = \ \int_{\Sigma_{0}\cap\left\{r\leq R\right\}}\left[\frac{1}{(r-M)}\cdot \Big(\partial_{r}(r\psi)\Big)^{2}\right]dr\,d\omega+\int_{\Sigma_{0}}\Big(J^{T}[\psi]\cdot \textbf{n}_{\Sigma_{0}}\Big)\, dg_{\Sigma_{0}}<\infty \]
since $\psi$ is a smooth function on $\Sigma_{0}$ supported on the set $\left\{M<\tilde{R_{1}}\leq r\leq \tilde{R_{2}}\right\}$. Clearly, \eqref{finite1} follows from \eqref{eee1} and \eqref{eee2}.

Arguing by contradiction, let us assume that we in fact have
\begin{equation}
\int_{\mathcal{A}}(\partial_{r}\partial_{r}\psi)^{2}\, dg_{\mathcal{A}} <\infty.
\label{contra1}
\end{equation}
Then, in view of \eqref{wephi4}, \eqref{finite1} and \eqref{contra1}, we have 
\begin{equation}
\begin{split}
\int_{\mathcal{A}}\Bigg(\partial_{r}\partial_{r}\Big((r-M)^{2}\cdot\partial_{r}\phi\Big)\Bigg)^{2}\, dg_{\mathcal{A}}<\infty.
\label{wephi5}
\end{split}
\end{equation}
Since
\[ \partial_{r}\Big((r-M)^{2}\cdot\partial_{r}\phi\Big)\sim (r-M) \]
in $\mathcal{A}$ , we can apply the Hardy inequality \eqref{3hardy} for $p=-2$ to get
\begin{equation}
\int_{\mathcal{A}}\frac{1}{(r-M)^{2}}\cdot\Bigg(\partial_{r}\Big((r-M)^{2}\cdot\partial_{r}\phi\Big)\Bigg)^{2}\, dg_{\mathcal{A}}\leq C \int_{\mathcal{A}}\Bigg(\partial_{r}\partial_{r}\Big((r-M)^{2}\cdot\partial_{r}\phi\Big)\Bigg)^{2}\, dg_{\mathcal{A}}
\label{hardycontra1}
\end{equation}
Since
\[ \Big((r-M)^{2}\cdot\partial_{r}\phi\Big)^{2}\sim (r-M)^{4} \]
in $\mathcal{A}$ , we can apply once again the Hardy inequality \eqref{3hardy} for $p=-4$ to get
\begin{equation}
\int_{\mathcal{A}}\frac{1}{(r-M)^{4}}\cdot \Big( (r-M)^{2}\cdot \partial_{r}\phi\Big)^{2} dg_{\mathcal{A}} \leq C\int_{\mathcal{A}}\frac{1}{(r-M)^{2}}\cdot\Bigg(\partial_{r}\Big((r-M)^{2}\cdot\partial_{r}\phi\Big)\Bigg)^{2}\, dg_{\mathcal{A}}
\label{hardycontra2}
\end{equation}
Hence, in view of the estimates \eqref{contra1}, \eqref{hardycontra1} and \eqref{hardycontra2} we get
\begin{equation}
\int_{\mathcal{A}}
 \big( \partial_{r}\phi\big)^{2} <\infty
\label{hardycontra22}
\end{equation}
On the other hand, by virtue of \eqref{consphi} we have that the conserved charge $H[\phi]\neq 0$ and hence by Theorem \ref{theo2} it follows that 
\begin{equation}
\int_{\mathcal{A}}
 \big( \partial_{r}\phi\big)^{2} =\infty
\label{hardycontra44}
\end{equation}
which of course contradicts \eqref{hardycontra22}. Hence, our assumption that the integral \eqref{contra1} is finite is wrong and this completes the proof for $k=2$.  

For $k\geq 3$ we simply argue by repeatedly using the  Hardy inequality \eqref{1hardy} to obtain
\[\int_{\mathcal{A}}\big(\partial_{r}\partial_{r}\psi\big)^{2}dg_{\mathcal{A}} \leq C_{k}\int_{\mathcal{A}}\big(\partial_{r}^{k}\psi\big)^{2}dg_{\mathcal{A}} \]
which implies that the integral on the left hand side must also be infinite. 

Clearly, \eqref{whattoshow} holds generically. Indeed, if we consider a solutions $\Psi$ to the wave equation such that 
\[\int_{\mathcal{A}}\big(\partial_{r}\partial_{r}\Psi\big)^{2}dg_{\mathcal{A}} <\infty \]
then \eqref{whattoshow} holds for $\Psi+\epsilon\psi$ for arbitrarily small $\epsilon$ and $\psi$ given by \eqref{psiconstruction} above. This completes the proof.

 \end{proof}

 \section{Relation with the stability theory of MOTS}
	 \label{sec:RelationWithTheoryOfMOTS}
	
	In this section we provide a connection of our findings to the stability theory of marginally outer trapped surfaces (MOTS).

	Our analysis shows that a necessary condition for the existence of a non-degenerate Morawetz estimate up to and including the event horizon for solutions $\psi$ to \eqref{we1} is the vanishing of the conserved charge
		\begin{equation}
		H[\psi]=\int_{S_{\tau}}\left(Y\psi+\frac{1}{M}\psi\right),
		\label{cha}
		\end{equation}
 where $S_{\tau}=\Sigma_{\tau}\cap\mathcal{H}$. If the conserved charge does not vanish then $\psi$ fails to satisfy a Morawetz estimate regardless of its degree of regularity. This is in stark contrast with the trapping effect on the photon sphere where Morawetz estimates hold as long as loss of regularity is allowed. 
The necessity of the vanishing of the charge \eqref{cha} on \textbf{each} section $S_{\tau}$ implies that a global trapping effect takes place on degenerate horizons. We shall interpret this in terms of the stability theory of the sections $S_{v}$, the latter seen as marginally outer trapped surfaces on extremal Reissner--Nordstr\"{o}m.

It was shown	in  \cite{aretakisglue}  that the conserved charge \eqref{cha} arises from the kernel of the elliptic operator
\[\mathcal{O}_{S_{\tau}}\psi=\lapp\psi+2\zeta^{\sharp}\cdot\nabb\psi+\left[2\divv\,\zeta^{\sharp}+\partial_{v}( tr\underline{\chi})+\frac{1}{2}( tr\underline{\chi} )( tr\chi)\right]\cdot\psi,\]
where $\zeta$ denotes the torsion of the section $S_{\tau}$ and $tr\chi,tr\underline{\chi}$ the outgoing and ingoing null mean expansions, respectively. This operator was introduced in \cite{aretakiselliptic} where it was shown that in the case of Killing horizons it reduces to 
\begin{equation}
\mathcal{O}_{S_{\tau}}\Psi=\lapp\Psi+\divv\big(2\Psi\cdot\zeta\big)+tr\underline{\chi}\cdot\kappa\cdot\Psi,
\label{operatorextremal}
\end{equation}
where  $\kappa$ is the surface gravity. 

In view of work of Mars\footnote{We acknowledge private communication with Marc Mars. } \cite{mars} the operator \eqref{operatorextremal} coincides with the \textit{stability operator} on the marginally outer trapped surface $S_{\tau}$. In the case of sub-extremal black holes we have that $tr\underline{\chi}<0$ and $\kappa>0$ and hence the principal eigenvalue of the operator  $\mathcal{O}_{S_{\tau}}$ must be strictly negative. Hence the whole spectrum of the operator $\mathcal{O}_{S_{\tau}}$ is strictly negative. This implies that all sections $S_{\tau}$ of sub-extremal event horizons are \textit{stable} as MOTS. 

On the other hand, in the extremal case we have $\kappa=0$ which implies that the principal eigenvalue of $\mathcal{O}_{S_{\tau}}$ is zero. Hence, there is a \textit{unique function} (up to a constant) that belongs in the kernel of $\mathcal{O}_{S_{\tau}}$ and all the other eigenfunctions correspond to strictly negative eigenvalues. Hence, the sections $S_{\tau}$ of are stable as MOTS in all but \textbf{exactly one} transversal perturbation with respect to which they are marginally stable. That is, there is a transversal perturbation $S_{f_{kernel}}$ of the degenerate horizon with respect to which the outgoing null mean expansion is stationary (i.e.~has a critical point).
\begin{figure}[H]
\begin{center}
\includegraphics[width=3.5in]{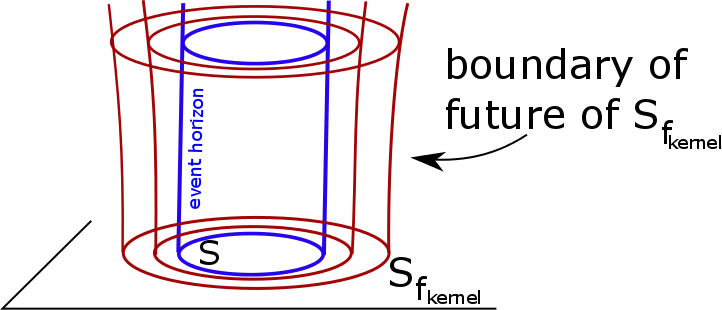}
\end{center}
\caption{The perturbation $S_{f_{kernel}}$ of the degenerate horizon with respect to which the outgoing null mean expansion is stationary.}
\end{figure}
 In other words, \textit{there is a unique perturbation of the degenerate horizon with respect to which the futures of the sections do \textbf{not} expand to \textbf{second} order}. Waves follow these non-expanding characteristic hypersurfaces leading to  the $L^{2}$-concentration of Theorem \ref{theo2}.

 Recalling that the conserved charge \eqref{cha} arises from the same perturbation (due to the kernel of $\mathcal{O}_{S_{\tau}}$), we conclude the failure of the expansion of the horizon in this direction induces the \textit{global} trapping effect.

\section{Acknowledgements}
\label{sec:Acknowledgements}

We would like to thank Mihalis Dafermos and Georgios Moschidis for several insightful discussions. The second author (S.A) acknowledges support through NSF grant DMS-1265538. The third author (D.G.) acknowledges support by the European Research Council grant no. ERC-2011-StG 279363-HiDGR.

Department of Mathematics, University of California, Los Angeles, CA 90095, United States, yannis@math.ucla.edu

\bigskip

Princeton University, Department of Mathematics, Fine Hall, Washington Road, Princeton, NJ 08544, United States, aretakis@math.princeton.edu

\bigskip

Department of Mathematics, University of Toronto Scarborough 1265 Military Trail, Toronto, ON, M1C 1A4, Canada, aretakis@math.toronto.edu

\bigskip

Department of Mathematics, University of Toronto, 40 St George Street, Toronto, ON, Canada, aretakis@math.toronto.edu

\bigskip

Department of Mathematics, Imperial College London, SW7 2AZ, London, United Kingdom, dejan.gajic@imperial.ac.uk

\bigskip

Department of Applied Mathematics and Theoretical Physics, University of Cambridge, Wilberforce Road, Cambridge CB3 0WA, United Kingdom, dg405@cam.ac.uk


\begin{thebibliography}{99}

\bibitem{blukerr}
{\sc Andersson, L., and Blue, P.}
\newblock Hidden symmetries and decay for the wave equation on the {K}err
  spacetime.
\newblock {\em arXiv:0908.2265\/} (2009).

\bibitem{yannis1}
{\sc Angelopoulos, Y.}
\newblock Nonlinear wave equations with null condition on extremal
  {R}eissner-{N}ordstr\"{o}m spacetimes {I}: Spherical symmetry.
\newblock {\em arXiv:1408.4478 (to appear in IMRN)\/} (2014).

\bibitem{SA10}
{\sc Aretakis, S.}
\newblock The wave equation on extreme {R}eissner--{N}ordstr\"om black hole
  spacetimes: stability and instability results.
\newblock {\em ar{X}iv:1006.0283\/} (2010).

\bibitem{aretakis1}
{\sc Aretakis, S.}
\newblock Stability and instability of extreme {R}eissner--{N}ordstr\"om black
  hole spacetimes for linear scalar perturbations {I}.
\newblock {\em Commun. Math. Phys. 307\/} (2011), 17--63.

\bibitem{aretakis2}
{\sc Aretakis, S.}
\newblock Stability and instability of extreme {R}eissner--{N}ordstr\"om black
  hole spacetimes for linear scalar perturbations {II}.
\newblock {\em Ann. Henri Poincar\'{e} 12\/} (2011), 1491--1538.

\bibitem{aretakis3}
{\sc Aretakis, S.}
\newblock Decay of axisymmetric solutions of the wave equation on extreme
  {K}err backgrounds.
\newblock {\em J. Funct. Analysis 263\/} (2012), 2770--2831.

\bibitem{aretakisglue}
{\sc Aretakis, S.}
\newblock The characteristic gluing problem and conservation laws for the wave
  equation on null hypersurfaces.
\newblock {\em arXiv:1310.1365\/} (2013).

\bibitem{aretakis2012}
{\sc Aretakis, S.}
\newblock A note on instabilities of extremal black holes from afar.
\newblock {\em Class. Quantum Grav. 30\/} (2013), 095010.

\bibitem{aretakiselliptic}
{\sc Aretakis, S.}
\newblock On a foliation-covariant elliptic operator on null hypersurfaces.
\newblock {\em to appear in IMRN, arXiv:1310.1348\/} (2013).

\bibitem{aretakis2013}
{\sc Aretakis, S.}
\newblock On a non-linear instability of extremal black holes.
\newblock {\em Phys. Rev. D 87\/} (2013), 084052.

\bibitem{aretakis4}
{\sc Aretakis, S.}
\newblock Horizon instability of extremal black holes.
\newblock {\em Adv. Theor. Math. Phys. 19\/} (2015), 507--530.

\bibitem{redshift}
{\sc Dafermos, M., and Rodnianski, I.}
\newblock The redshift effect and radiation decay on black hole spacetimes.
\newblock {\em Comm. Pure Appl. Math. 62\/} (2009), 859--919, ar{X}iv:0512.119.

\bibitem{lecturesMD}
{\sc Dafermos, M., and Rodnianski, I.}
\newblock Lectures on black holes and linear waves.
\newblock {\em in Evolution equations, {C}lay {M}athematics {P}roceedings,
  {V}ol. 17, Amer. Math. Soc., Providence, RI,\/} (2013), 97--205,
  ar{X}iv:0811.0354.

\bibitem{part3}
{\sc Dafermos, M., Rodnianski, I., and Shlapentokh-Rothman, Y.}
\newblock Decay for solutions of the wave equation on {K}err exterior
  spacetimes {III: The full subextremal case} $|a| < m$.
\newblock {\em arXiv:1402.7034\/}.

\bibitem{newdaf}
{\sc Dafermos, M., G. Holzegel and Rodnianski, I.}
\newblock The linear stability of the Schwarzschild solution
to gravitational perturbations
\newblock {\em arXiv:1601.06467\/} (2016).

\bibitem{dd2012}
{\sc Dain, S., and Dotti, G.}
\newblock The wave equation on the extreme {R}eissner--{N}ordstr\"om black
  hole.
\newblock {\em ar{X}iv:1209.0213\/} (2012).

\bibitem{semyon2}
{\sc Dyatlov, S.}
\newblock Exponential energy decay for {K}err--de {S}itter black holes beyond
  event horizons.
\newblock {\em Math. Research Letters 18\/} (2011), 1023--1035.

\bibitem{gajic}
{\sc Gajic, D.}
\newblock Linear waves in the interior of extremal black holes {I}.
\newblock {\em arXiv:1509.06568\/} (2015).

\bibitem{gusmu1}
{\sc Holzegel, G., and Smulevici, J.}
\newblock Decay properties of {K}lein--{G}ordon fields on {K}err--{A}d{S}
  spacetimes.
\newblock {\em Comm. Pure Appl. Math. 66\/} (2013), 1751--1802.

\bibitem{Kay1987}
{\sc Kay, B., and Wald, R.}
\newblock Linear stability of {S}chwarzschild under perturbations which are
  nonvanishing on the bifurcation 2-sphere.
\newblock {\em Class. Quantum Grav. 4\/} (1987), 893--898.

\bibitem{keir}
{\sc Keir, J.}
\newblock Slowly decaying waves on spherically symmetric spacetimes and an
  instability of ultracompact neutron stars.
\newblock {\em arXiv:1404.7036\/} (2014).

\bibitem{muchT}
{\sc Klainerman, S.}
\newblock Uniform decay estimates and the {L}orentz invariance of the classical
  wave equation.
\newblock {\em Comm. Pure Appl. Math. 38\/} (1985), 321--332.

\bibitem{hm2012}
{\sc Lucietti, J., Murata, K., Reall, H.~S., and Tanahashi, N.}
\newblock On the horizon instability of an extreme {R}eissner--{N}ordstr\"om
  black hole.
\newblock {\em JHEP 1303\/} (2013), 035, arXiv:1212.2557.

\bibitem{hj2012}
{\sc Lucietti, J., and Reall, H.}
\newblock Gravitational instability of an extreme {K}err black hole.
\newblock {\em Phys. Rev. D86:104030\/} (2012).

\bibitem{mars}
{\sc Mars, M.}
\newblock Stability of {MOTS} in totally geodesic null horizons.
\newblock {\em Class. Quantum Grav. 29}, 14 (2012), DOI: 10.1088.

\bibitem{molog}
{\sc Moschidis, G.}
\newblock Logarithmic local energy decay for scalar waves on a general class of
  asymptotically flat spacetimes.
\newblock {\em arXiv:1509.08495\/} (2015).

\bibitem{murata2012}
{\sc Murata, K.}
\newblock Instability of higher dimensional extreme black holes.
\newblock {\em Class. Quantum Grav. 30\/} (2013), 075002.

\bibitem{harvey2013}
{\sc Murata, K., Reall, H.~S., and Tanahashi, N.}
\newblock What happens at the horizon(s) of an extreme black hole?
\newblock {\em arXiv:1307.6800\/} (2013).

\bibitem{ori2013}
{\sc Ori, A.}
\newblock Late-time tails in extremal {R}eissner-{N}ordstr\"{o}m spacetime.
\newblock {\em arXiv:1305.1564\/} (2013).

\bibitem{ralston1}
{\sc Ralston, J.}
\newblock Solutions of the wave equation with localized energy.
\newblock {\em Comm. Pure Appl. Math. 22\/} (1969), 807--823.

\bibitem{regge}
{\sc Regge, T., and Wheeler, J.}
\newblock Stability of a {S}chwarzschild singularity.
\newblock {\em Phys. Rev. 108\/} (1957), 1063--1069.

\bibitem{janpaper}
{\sc Sbierski, J.}
\newblock Characterisation of the energy of {G}aussian beams on {L}orentzian
  manifolds with applications to black hole spacetimes.
\newblock {\em arXiv:1311.2477\/} (2013).

\bibitem{sela}
{\sc Sela, O.}
\newblock Late-time decay of perturbations outside extremal charged black hole.
\newblock {\em arXiv:1510.06169\/} (2015).

\bibitem{tataru2}
{\sc Tataru, D., and Tohaneanu, M.}
\newblock A local energy estimate on {K}err black hole backgrounds.
\newblock {\em Int. Math. Res. Not. 2011\/} (2008), 248--292.

\bibitem{iapwnes}
{\sc Tsukamoto, N., Kimura, M., and Harada, T.}
\newblock High energy collision of particles in the vicinity of extremal black
  holes in higher dimensions: {B}anados-{S}ilk-{W}est process as linear
  instability of extremal black holes.
\newblock {\em arXiv:1310.5716\/} (2013).

\bibitem{jaredw}
{\sc Wunsch, J., and Zworski, M.}
\newblock Resolvent estimates for normally hyperbolic trapped sets.
\newblock {\em Ann. Henri Poincar\'{e} 12\/} (2011), 1349--1385.

\end{thebibliography}
\end{document}